\pdfoutput=1
\documentclass[sigconf,nonacm,balance=false]{acmart}

\usepackage{preamble}

\begin{document}

\title{\papertitle}

\author{Jamshed Khan}
\affiliation{%
  \institution{Northeastern University}
  \city{Boston}
  \state{Massachusetts}
  \country{USA}}

\author{Prashant Pandey}
\affiliation{%
  \institution{Northeastern University}
  \city{Boston}
  \state{Massachusetts}
  \country{USA}}

\begin{abstract}

Suffix arrays power deduplication, contamination testing, and training data
attribution over the corpora used to train large language models, and must be
rebuilt for every new corpus or tokenizer. However,
constructing them at scale is slow.
\ig's
indexing pipeline takes about two days to process 1.4~trillion tokens on its
authors' 128-core CPU server. GPU-accelerated constructors are an order of
magnitude faster, but their working sets must fit in GPU memory. Partitioning
the corpus only moves this limit to host memory. Sorting the suffixes of one
partition needs the ranks of suffixes in other
partitions, which existing designs keep in a corpus-wide rank table that must
stay in memory.

We present \gpusaca, a GPU-accelerated out-of-core suffix-array constructor.
Its partitions each own a fixed, contiguous range of the final suffix array, so
the GPU sorts them independently and never merges them. \gpusaca\ then replaces
the rank table with \emph{streamed rank exchange}. Ranks are sent directly to
the partitions that need them, and an invariant makes partition-local ranks
sufficient for every comparison. With no rank table, intermediate data spills
across GPU memory, host memory, and SSD as the corpus grows, and host memory only
needs to hold the text.

On seven corpora of 8--238 billion tokens, \gpusaca\ builds suffix arrays
on a single GPU server $3.0$--$8.2\times$ faster end to end than \ig's
pipeline, and its construction is $6.3$--$11.4\times$ faster than
\libsais's. It indexes the 198-billion-token \data{c4} corpus and the
238-billion-token \data{pile} in 23 and 37 minutes end to end, both of which exceed
\libsais's host-memory capacity on the same server. A batched query engine
over the index reproduces \ig's counts at median speedups of
\numrange{10.8}{32.6}$\times$, with the index in GPU memory, host memory, or
on SSD.

\end{abstract}
 
\ccsdesc[500]{Theory of computation~Sorting and searching}
\ccsdesc[500]{Information systems~Data structures}

\keywords{suffix array, parallel algorithms, GPU, systems, out-of-core, external memory, memory hierarchy}

\maketitle

\section{Introduction}
\label{sec:introduction}

A \emph{suffix array} (SA)~\cite{manber90} lists every position in a tokenized
corpus in lexicographic order of the text starting there, turning sequence
queries into binary searches. Suffix arrays have become core infrastructure for
working with LLM training data. Deduplication, contamination and membership
testing, $n$-gram statistics, retrieval, and training data attribution all
reduce to such searches, in tools such as $\infty$-gram language
models~\cite{liu24}, OLMoTrace~\cite{liu25}, and semantic
search~\cite{yoneda26}.

These tools need a whole-corpus index before they can run, and each new crawl,
filter, or tokenizer requires rebuilding it. This makes construction time part
of the data pipeline. Even exact-substring deduplication relies on suffix
arrays~\cite{lee22,kandpal22}, and sharding leaves cross-shard duplicates
undetected while requiring every query to search every shard. Queries also
arrive in bulk, e.g., attribution checks every span of a model output, while
retrieval-based speculative decoding searches the corpus during
generation~\cite{he24,ho24}. Fast construction and search therefore let these
tools keep pace with changing corpora and workloads.

Building the suffix array is expensive. At two bytes per token and five bytes
per suffix, a trillion-token corpus requires \qty{2}{\tera\byte} of text and a
\qty{5}{\tera\byte} index. \ig~\cite{liu24}, a suffix-array search engine
for LLM corpora, builds its indexes with a CPU pipeline that suffix-sorts
blocks of the text and merges them. Its authors report about two days to
index the \num{1.4} trillion tokens of RedPajama on a 128-core CPU server with
\qty{1}{\tera\byte} of memory, split into separate shards, each with its own
index. Compressed alternatives, such as the FM-index of
infini-gram mini~\cite{xu25}, shrink the index, but are built from the same
suffix order.

Existing suffix-array constructors trade speed for scale, and what each must
keep in memory sets its limit (Table~\ref{tab:positioning}). In-memory CPU
constructors, such as the multithreaded SA-IS implementation
\libsais~\cite{libsais}, are limited by host memory. GPU constructors, such as
\libcubwt~\cite{libcubwt}, are about an order of magnitude faster, but keep
their entire working set in the memory of one GPU~\cite{libcubwt} or several
GPUs~\cite{buren19}, making GPU memory their limit. Out-of-core CPU
constructors scale further, but at CPU speed, sorting or merging their data on
disk~\cite{dementiev08,barsky20,bingmann16,karkkainen15,lee22}. GPU data
processing beyond GPU memory faces a similar limit more broadly. Joins,
graph systems, and vector search retain some complete structure in GPU or host
memory, and stop scaling when that structure no longer
fits~\cite{maltenberger25,lutz20,lutz22,min20,wang23,peng26}.

\begin{table}[t]
\caption{Suffix-array constructors, the data they must keep resident, and the
largest corpus this allows on our server, with $M = \qty{1}{\tera\byte}$ of
host memory and \qty{96}{\giga\byte} of GPU memory (\S\ref{sec:prelim-hw}).}
\label{tab:positioning}
\footnotesize
\setlength{\tabcolsep}{3pt}
\begin{tabular}{@{}>{\raggedright\arraybackslash}p{0.33\columnwidth}>{\raggedright\arraybackslash}p{0.43\columnwidth}>{\raggedright\arraybackslash}p{0.2\columnwidth}@{}}
    \toprule
    Constructor family & Must stay resident & Wall (tokens) \\
    \midrule
    In-memory CPU~\cite{libsais} & text, SA, workspace (host) & $\approx 100$B \\
    GPU~\cite{libcubwt,buren19} & text, working set, rank table (GPU) & \num{1.6}B \\
    Out-of-core CPU~\cite{dementiev08,karkkainen15,bingmann16,barsky20} & none; global sort or merge per round & none \\
    Sharded out-of-core CPU~\cite{liu24} & one shard's text (host) & $\approx M / 2$ per shard \\
    \midrule
    \gpusaca & text or partition map (host) & $\approx M / 2$ \\
    \bottomrule
\end{tabular}
\end{table}

\para{Challenge}
A natural solution is to partition the corpus and sort each partition on the
GPU. Partitioning alone, however, is insufficient for \emph{prefix
doubling}~\cite{manber90,osipov12,wang16,buren19}, the method used by GPU
constructors. Prefix doubling sorts suffixes in rounds, which we call
\emph{generations}, comparing prefixes twice as long in each successive one.
After each generation, every suffix has a
\emph{rank} indicating its position in the order established so far. Suffixes
that remain tied share a rank. To double the prefix length, each suffix pairs
its own rank with that of the suffix $c$ positions ahead, where $c$ is the
current prefix length. This avoids reading the text itself, but the second
suffix may lie in any partition. Existing designs therefore
maintain a \emph{rank table}~\cite{manber90} that records the current rank of
every suffix in the corpus. Each generation reads this table at scattered
positions and rewrites it, requiring the entire table to remain in memory.
Partitioning thus replaces the GPU memory limit with a host memory limit. On
our server, storing a rank table alongside the text would cap the corpus at
about 130B tokens, below \data{c4} and \data{pile}
(\S\ref{sec:evaluation-reach}).

\para{Key insight}
Our key observation is that each rank has at most one reader: in the next
generation, only the suffix $c$ positions before it needs that rank. The rank can
therefore be sent directly to its reader rather than stored for lookup.
External-memory and distributed-memory constructors also transmit ranks to the
suffixes that need them~\cite{dementiev08,barsky20,kulla07,flick15}, but by
sorting or merging all their data globally in every round, over ranks numbered
across the whole corpus. \gpusaca\ needs neither: each rank goes to its
reader's partition, an invariant keeps every rank comparison inside one
partition, and ranks travel as streams that any memory tier can hold. We call
this \emph{streamed rank exchange}.

\para{This paper}
\gpusaca\ builds the suffix array in three steps (\S\ref{sec:method}), and adds a
query engine over the result. It first partitions the suffixes by output range.
Splitters sampled from the corpus's own $K$-grams, as in
samplesort~\cite{frazer70}, fix each partition's range of the final suffix array
before sorting begins (\S\ref{sec:method-partitioned}).
The GPU sorts and ranks each partition separately,
without ever merging partitions. GPU memory thus limits the partition size
rather than the corpus size. Some suffix sorters also partition by output
order~\cite{ghoting09,mansour11,karkkainen07,khan23}, but sort each partition
against the text, which they must keep in memory or reread for each batch of
partitions. Prefix doubling instead reads the text only in its first
generation, and resolves long repeats in logarithmically many generations.

Partitions then exchange ranks as messages rather than reading them from a
shared rank table (\S\ref{sec:method-exchange}). After each generation, every
partition sends each rank to the partition containing its reader, using a
compact partition map to identify the destination.
This is similar to push-based graph
processing~\cite{roy13, beamer12}. Two properties make this exchange correct.
First, our \emph{target coherence invariant} guarantees that any two suffixes
that remain tied within a partition receive their next ranks from the same
partition, so partition-local ranks suffice. Without it, ranks would need
corpus-wide numbering, which widens every rank, message, and sort key as the
corpus grows. Second, a suffix retires and stops sending once its rank is final,
even though a later generation may still need that rank.
The suffix that
received it last forwards it as a \emph{poison pill}, and retires in turn. Every
needed rank thus arrives, exactly once.

Eliminating the rank table leaves only two structures read at random,
the text and the partition map (\S\ref{sec:method-residency}). Every other
structure is a stream of fixed-size chunks, placed in GPU memory, host
memory, or on SSD. As in the semi-external and external-memory
models~\cite{abello02, aggarwal88}, what must remain in memory determines the
largest corpus, while what may remain in memory determines only the speed.
Placing each chunk as it is written in the fastest tier with room, and evicting
the records of retired suffixes first, minimizes SSD traffic, with no
lookahead. \gpusaca\ keeps the text and the partition map in host memory, only
one at a time when memory is short, so host memory $M$ bounds it at a little
under $M / 2$ tokens: the per-shard reach of \ig's pipeline, at GPU speed.

\begin{figure}[t]
    \centering
    \includegraphics{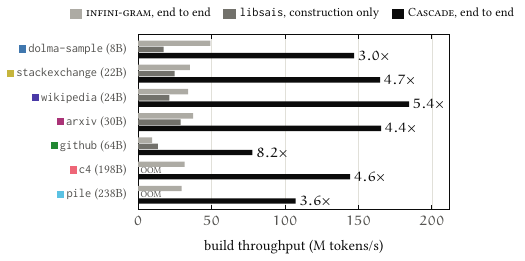}
    \caption{
        Build throughput per corpus. Labels give \gpusaca's end-to-end
        speedup over \ig; \libsais\ is timed on construction alone, its most
        favorable measure (Table~\ref{tab:baseline}).
    }
    \label{fig:intro-build}
\end{figure}

\para{Results}
Figure~\ref{fig:intro-build} summarizes our build results. We evaluate \gpusaca\ on
seven corpora containing 8 to 238 billion tokens. It builds the index
\numrange{3.0}{8.2}$\times$ faster end to end than \ig's pipeline, and its
construction is \numrange{6.3}{11.4}$\times$ faster than that of \libsais.
Unlike \ig's merge, it never truncates suffix comparisons
(\S\ref{sec:evaluation-setup}). On prefixes of \data{c4}, its lead over \ig\ doubles
as the corpus grows from 8B to 64B tokens, and holds beyond
(\S\ref{sec:evaluation-build}). \gpusaca\ indexes \data{c4} and \data{pile} in 23 and 37
minutes end to end, while \libsais\ runs out of host memory on both.
\libcubwt\ is limited to byte alphabets and GPU memory, and builds none of the
corpora.
On our server, host memory would bound \gpusaca\ at about 420B tokens, against
about 130B with a rank table.

Throughput degrades gradually as memory shrinks (\S\ref{sec:evaluation-reach}).
Using only \qty{5}{\giga\byte} of GPU memory, the \data{arxiv} and
\data{github} builds take 1.3$\times$ and 1.5$\times$ as long.
With host memory cut to its floor,
the text alone, they take 1.7$\times$ and 3.3$\times$ as long.

\gpusaca's batched query engine reproduces \ig's counts on every workload, at median
speedups over \ig's engine of \numrange{15.7}{32.6}$\times$ with the index in
GPU or host memory, and \numrange{10.8}{15.2}$\times$ with it on SSD
(\S\ref{sec:evaluation-query}).

This paper makes four contributions.
\begin{compactenum}
    \item \textbf{Partitioned prefix doubling}
      (\S\ref{sec:method-partitioned}). Sampled splitters fix each partition's
      output range, so the GPU sorts partitions independently and never merges
      them. To our knowledge, \gpusaca\ is the first GPU suffix-array
      constructor for general text whose input can exceed the memory of its
      GPUs.

    \item \textbf{Streamed rank exchange}
      (\S\ref{sec:method-exchange}). Rank messages, poison pills, and the
      target coherence invariant replace the rank table, with partition-local
      ranks of fixed width.

    \item \textbf{Tiered residency across the memory hierarchy}
      (\S\ref{sec:method-residency}). Intermediate data streams across GPU
      memory, host memory, and SSD, with
      minimal SSD traffic and no throughput cliff.

    \item \textbf{Queries over the index}
      (\S\ref{sec:method-query}, \S\ref{sec:system-query}).
      A small sample of the
      index, kept in GPU memory, narrows each query pattern of a batch to a bracket of
      the suffix array, which the engine searches in GPU memory, host memory, or on SSD.
\end{compactenum}

\section{Preliminaries}
\label{sec:prelim}

This section defines suffix arrays and prefix doubling, the method that \gpusaca\ builds on.
It then describes the hardware that \gpusaca\ targets.

\subsection{Suffix arrays}
\label{sec:prelim-text}

Let $T[0, n)$ be a text of length $n$ over a totally ordered alphabet $\Sigma$.
Its last symbol $T[n - 1] = \sentinel$ is a \defn{sentinel}: it occurs nowhere else, and it is smaller than every other symbol.~\footnote{
    Real corpora carry no such symbol; see Appendix~\ref{sec:supp-terminator}.
}
The suffix $\suf{j} = T[j, n)$ starts at position $j$, and we identify each suffix with its starting position $j$.
The \defn{$c$-context} of suffix $j$ is its prefix of length $c$, $\ctx{j}{c}$.
The \defn{longest common prefix} $\lcp(\suf{i}, \suf{j})$ is the length of the longest prefix that suffixes $i$ and $j$ share.
Each suffix has a unique length, so all suffixes are distinct.
Since \sentinel\ is unique, $\lcp(\suf{i}, \suf{j}) < \min(|\suf{i}|, |\suf{j}|)$ for $i \neq j$, and a context that contains \sentinel\ occurs at exactly one position in $T$.

The \defn{suffix array} $\sa[0, n)$ lists the suffixes in lexicographic order.
The \defn{rank} of suffix $j$ is its position $r$ in \sa, where $\sa[r] = j$.
The \defn{rank table} maps every suffix to its rank.
It is the inverse of \sa, also known as the inverse suffix array~\cite{manber90}.
The \defn{repeat depth} $\rdep(j) = \max_{i \neq j} \lcp(\suf{i}, \suf{j})$ is the length of the longest prefix that $j$ shares with another suffix.
It is how far a comparison must read to tell $j$ apart from every other suffix.

\subsection{Prefix doubling}
\label{sec:prelim-doubling}

\begin{figure*}[t]
    \centering
    \includegraphics{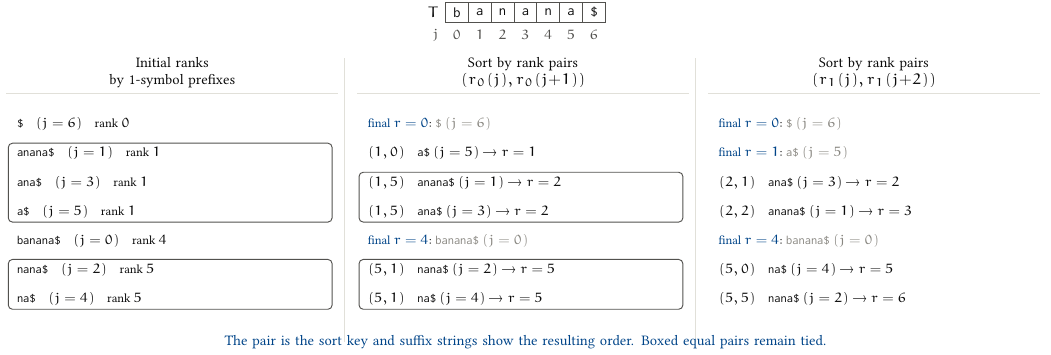}
    \caption{
        Prefix doubling on $T = \texttt{banana\$}$.
        Each column shows the suffixes in the order established by that round. Pairing the one-symbol ranks at $j$ and $j{+}1$ resolves two-symbol prefixes, and the next round pairs two-symbol ranks at $j$ and $j{+}2$, combining the prefixes $T[j,j{+}2)$ and $T[j{+}2,j{+}4)$ to resolve four-symbol prefixes.
        Boxes mark equal keys. The two suffixes with prefix ``an'' and the two with prefix ``na'' remain tied after the first pair, then split after the second.
        Each later round sorts rank pairs, not suffix strings. Suffixes are shown in the order those pairs establish.
    }
    \label{fig:prefix-doubling}
\end{figure*}

Prefix doubling~\cite{manber90} sorts the suffixes in rounds, by contexts of doubling length (Figure~\ref{fig:prefix-doubling}).
After a round that sorts by $c$-contexts, the rank of $j$ is the number of suffixes whose $c$-context is strictly smaller than that of $j$.
Suffixes with equal $c$-contexts therefore share a rank.
The $2c$-context of $j$ is its $c$-context followed by the $c$-context of $j + c$.
So sorting the suffixes by the pair of ranks of $j$ and $j + c$, and ranking them again, resolves their $2c$-contexts.
Figure~\ref{fig:prefix-doubling} traces this on $T = \texttt{banana\$}$.
After the first pairing, suffixes 1 and 3 still share the pair $(1, 5)$, since both start with \texttt{an}.
The next round pairs their ranks with those of suffixes 3 and 5, giving $(2, 2)$ and $(2, 1)$, and so places \texttt{ana\$} before \texttt{anana\$}.
We call each round a \defn{generation}, and $j + c$ the \defn{extension suffix} of $j$ in it.
The textbook implementation keeps the current rank of every suffix in the rank table, and reads the rank of each extension suffix from it.

A suffix whose rank is unique is final: its position in \sa\ is known, and it \defn{retires} from all later generations~\cite{larsson07}.
Only the suffixes that still share a rank carry forward.
Suffix $j$ thus takes part in at most $1 + \lfloor \log_2 \rdep(j) \rfloor$ generations after the first, and only in the first if $\rdep(j) = 0$.
The total work therefore depends on the distribution of repeat depths, and not only on $n$.

\para{Initial contexts}
\Gen{0} sorts the suffixes by their $\ell$-contexts, for a given $\ell \geq 1$, rather than by single symbols.
After \gen{t}, contexts of length $\cx{t} = 2^{t} \ell$ are resolved.
Generation~$t \geq 1$ therefore pairs $j$ with $j + \cx{t - 1}$.
Suffix $j$ takes part in at most $1 + \lfloor \log_2(\rdep(j) / \ell) \rfloor$ of these \defn{doubling} generations, and in none if $\rdep(j) < \ell$.

\subsection{Hardware model}
\label{sec:prelim-hw}

We target a single server in which CPUs, the \defn{host}, drive a GPU, the \defn{device}, over PCIe, with local NVMe SSDs.
Our server has an NVIDIA RTX PRO 6000 Blackwell GPU on PCIe~5.0~$\times$16, two 64-core Intel Xeon 6710E CPUs, and two Samsung PM9A3 NVMe SSDs; Table~\ref{tab:hardware} lists its measured rates.
Its three tiers are \defn{device memory}, \defn{host memory}, and the SSDs.
Each tier holds about ten times as much as the tier above it, but serves random reads at a fraction of that tier's rate.
Device memory holds \qty{96}{\giga\byte}, streams \qty{1.5}{\tera\byte\per\second}, and serves 6.5 to 15 billion random 8-byte reads per second.
Host memory holds \qty{1}{\tera\byte} and streams \qty{530}{\giga\byte\per\second}, and all 128 cores together serve 2.1 billion random reads per second.
Each of the two \qty{3.84}{\tera\byte} SSDs reads \qty{6.7}{\giga\byte\per\second} and writes \qty{4.1}{\giga\byte\per\second}, but serves only \num{870000} random 4~KB reads per second.
\gpusaca\ places its data across these three tiers (\S\ref{sec:method-residency}).

The host launches \defn{kernels} on the device: functions that the device runs on many threads at once.
The host also copies data between host and device memory, at about \qty{57}{\giga\byte\per\second} in each direction.
These copies reach that rate only from \defn{pinned} (page-locked) host memory.
Pinning is slow, so buffers are pinned once and reused.
Kernels and copies are issued on \defn{CUDA streams}, queues that each run in order.
Work on different streams overlaps, so the device can compute on one batch while it copies another.
The device can also read pinned host memory directly.
Such random reads reach only 25 to 94 million per second, however, because each one waits for an address translation at the host.
The host's cores serve random reads at least 20 times faster.

File I/O normally goes through the operating system's \defn{page cache}, which takes up host memory.
\defn{Direct I/O} (\odirect) bypasses the cache with block-aligned transfers.
A memory-mapped file instead faults its pages in on demand.

\section{Prefix Doubling Without a Rank Table}
\label{sec:method}

This section describes how \gpusaca\ constructs and searches a suffix array when the text is too large to fit in device memory.
Construction has three steps.
First, \gpusaca\ partitions suffixes by their initial prefixes, which assigns each partition a contiguous range of the final suffix array (\S\ref{sec:method-partitioned}).
It then sorts each partition by prefix doubling (\S\ref{sec:method-doubling}), exchanging the ranks needed between rounds as messages rather than storing a rank table (\S\ref{sec:method-exchange}).
Finally, it streams intermediate data through device memory, host memory, and SSD (\S\ref{sec:method-residency}).
The resulting index can be searched wherever it resides (\S\ref{sec:method-query}).
Algorithm~\ref{alg:construct} summarizes the construction, and Appendix~\ref{sec:supp-example} traces it on a small example.

\begin{algorithm}[t]
\caption{
    The construction of the suffix array of $T$: partitioning (lines~1--2), \gen{0} (lines~3--5), the doubling generations, one barrier apart (lines~6--15), and the output (lines~16--17).
    Sending or forwarding $\langle y, r, \mathit{final} \rangle$ appends it to the inbox of partition $\pmap(y)$. The bit \textit{final} is set if $r$ is final, and on every forward.
    No message goes to a position $y < 0$, nor, by the liveness filter, to a suffix known to have retired.
}
\label{alg:construct}
\begin{codebox}
    \Procname{$\gpusaca(T)$}
    \li sample the $K$-grams of $T$; choose splitters; compute \pmap
    \li scatter each suffix $j$ to partition $\pmap(j)$
    \li \For each partition \algcmnt{\gen{0}}
    \li     \Do sort by key~\eqref{eq:key0}; assign the ranks \lrk{0}
    \li         send $\langle j - \ell, \lrk{0}(j), \mathit{final} \rangle$ for each suffix $j$
            \End
    \li release $T$; $t \gets 1$
    \li \While some suffix is active \algcmnt{barrier~$t$}
    \li     \Do \For each partition with active suffixes \algcmnt{any order}
    \li         \Do join its inbox to its active suffixes
    \li             sort them by key~\eqref{eq:keyt}; assign the ranks \lrk{t}
    \li             \For each of them, $x$
    \li                 \Do send $\langle x - \cx{t}, \lrk{t}(x), \mathit{final} \rangle$
    \li                     \If $x$ received a final rank $r$
    \li                         \Then forward $\langle x - \cx{t - 1}, r, \mathit{final} \rangle$
                            \End
                        \End
                    \End
    \li     $t \gets t + 1$
            \End
    \li \For each suffix $j$: $\sa[\text{offset of } \pmap(j) + \lrk{}(j)] \gets j$
    \li \Return \sa
\end{codebox}
\end{algorithm}

\subsection{Partitioning by output range}
\label{sec:method-partitioned}
\label{sec:method-partition}

The partitioning rule serves two purposes: it limits the number of suffixes each device sort must handle, and it puts the partitions in suffix-array order before they are sorted. Thus, each partition can later write directly to its own contiguous range of the output.

\para{Splitters}
Fix a splitter length $K$ and a target partition size $B$.
As in samplesort~\cite{frazer70}, \gpusaca\ samples about $128 \log_2 n$ positions of $T$ per partition, uniformly at random, and sorts the $K$-grams that start there.
A sweep over the sorted sample closes a partition just before its estimated size would exceed $B$, and moves the boundary back so that equal $K$-grams never straddle it.
The $K$-gram at each boundary becomes a \defn{splitter}, and the splitters cut the $K$-grams into $p \ge \lceil n / B \rceil$ ranges.
The \defn{partition map} \pmap\ sends each suffix $j$ to the range that its first $K$ symbols fall in, and partition $\pmap(j)$ holds it.

\para{Each partition owns a range of the output}
The splitters order partitions by their $K$-grams. Because equal $K$-grams stay together, every suffix in partition $i$ precedes every suffix in a later partition.
Partition $i$ therefore fills exactly the range of \sa\ that starts at its \defn{offset}, the number of suffixes in partitions $0$ to $i - 1$.
Partitions are never required to be merged. Each partition is sorted and ranked independently, and writes its final order straight to its range of the output.
As a result, the device memory, not $n$, sets the partition size, and only the partition count $p$ grows with $n$.

\para{Feasibility}
Oversampling splitters keeps every partition within a constant factor of $B$ with high probability~\cite{blelloch91}.
Suffixes that share a $K$-gram, however, always land in the same partition, so the most frequent $K$-gram of $T$ must fit in one partition.
\gpusaca\ checks this on the sample, and aborts before any sorting if it fails.
A longer $K$ splits the heaviest $K$-grams (\S\ref{sec:evaluation-reach}), and every corpus we index passes with a wide margin (\S\ref{sec:evaluation-structure}).

\subsection{Prefix doubling within partitions}
\label{sec:method-doubling}

Within each partition, \gpusaca\ runs prefix doubling (\S\ref{sec:prelim-doubling}) on \defn{local ranks}: the rank $\lrk{t}(j)$ of suffix $j$ after \gen{t} counts the suffixes of its own partition that precede it in the order resolved so far.
\Gen{0} sorts each partition by the key
\begin{equation}
\label{eq:key0}
    \bigl( \ctx{j}{\ell},\; \pmap(j + \ell) \bigr),
\end{equation}
the $\ell$-context of $j$ and the partition of its extension suffix in \gen{1}. \S\ref{sec:method-tci} explains why the second component is needed.
$\lrk{0}(j)$ is the number of suffixes in the partition of $j$ with strictly smaller keys.
A suffix whose rank is unique in its partition retires, and keeps its rank, which is final.
We call a suffix \defn{active} in a generation if that generation ranks it. Every suffix is active in \gen{0}, and a suffix is active in \gen{t \ge 1} if it was still tied after \gen{t - 1}.
\Gen{t \ge 1} sorts the active suffixes of each partition by the key
\begin{equation}
\label{eq:keyt}
    \bigl( \lrk{t - 1}(j),\; \lrk{t - 1}(j + \cx{t - 1}) \bigr),
\end{equation}
the pair of the rank of $j$ and the rank of its extension suffix, its \defn{extension rank}.
Each tied suffix's new rank is its old rank plus the number of suffixes tied with it whose extension ranks are strictly smaller.
Once no suffix is active, every rank is final, and suffix $j$ goes to position $\lrk{}(j)$ past the offset of its partition in \sa.

\para{The cross-partition dependency}
Key~\eqref{eq:keyt} needs two ranks: the rank of suffix $j$, and the rank of the suffix starting at $j + \cx{t - 1}$.
These ranks may belong to different partitions. Since local ranks are meaningful only within their own partition, a partition cannot always use the extension rank directly.
Figure~\ref{fig:prefix-doubling} shows the problem for \texttt{banana\$}, with $\ell = K = 1$ and one partition per first symbol.
Partition \texttt{a} contains suffixes 1, 3, and 5; partition \texttt{n} contains suffixes 2 and 4.
In the first doubling round, suffix 1 needs the rank of suffix 2, which belongs to partition \texttt{n}.
The figure shows global ranks.
If \gen{0} ranked each partition by the $\ell$-context alone, suffixes 2, 4, and 6 would all have local rank 0 in their partitions, and suffixes 1, 3, and 5 would all receive key $(0, 0)$, even though \texttt{a\$} precedes \texttt{an}.
The rank table avoids this ambiguity, but requires a random-access rank for every suffix. It takes $\lceil \log_2 n \rceil$ bits per suffix and is rewritten in every generation.
\gpusaca\ resolves the ambiguity with target coherence (\S\ref{sec:method-tci}), and replaces the table with messages (\S\ref{sec:method-message}).

\subsection{Streamed rank exchange}
\label{sec:method-exchange}

\subsubsection{\textbf{Target coherence}}
\label{sec:method-tci}

Local ranks of different partitions never need to be compared if suffixes tied within a partition always have their extension suffixes in one partition.
We call this property \defn{target coherence}.
The component $\pmap(j + \ell)$ of key~\eqref{eq:key0} establishes it for \gen{1}, and Theorem~\ref{thm:tci} shows that every later generation preserves it.
Thus, suffixes still tied within one partition always need extension ranks from the same partition. The algorithm can compare those local ranks without converting them to global ranks.

\begin{theorem}[Target coherence]
\label{thm:tci}
After every generation~$t$, any two suffixes tied within a partition share their first \cx{t} symbols, and their extension suffixes in \gen{t + 1} lie in one partition.
\end{theorem}
\begin{proof}[Proof sketch]
By induction on $t$.
Key~\eqref{eq:key0} contains $\pmap(j + \ell)$, so suffixes tied after \gen{0} agree on the partition of their extension suffixes.
For $t \ge 1$, suffixes $j$ and $j'$ tied after \gen{t} were tied after \gen{t - 1}, and so were their extension suffixes $m = j + \cx{t - 1}$ and $m' = j' + \cx{t - 1}$, which lie in one partition by induction.
The hypothesis for both pairs gives $j$ and $j'$ equal \cx{t}-contexts, and places $m + \cx{t - 1} = j + \cx{t}$ and $m' + \cx{t - 1} = j' + \cx{t}$ in one partition.
Appendix~\ref{sec:supp-proof-tci} gives the full proof.
\end{proof}

\para{Why \gen{0} needs the partition}
Without the component $\pmap(j + \ell)$, two suffixes tied after \gen{0} could have their extension suffixes in different partitions, whose local ranks cannot be compared, and \gen{1} could order them wrongly.
With it, the property carries through the ranks, so later keys need no such component.
In the example of \S\ref{sec:method-doubling}, \gen{0} gives suffix 5 the key $(\texttt{a}, \pmap(6))$, and suffixes 1 and 3 the keys $(\texttt{a}, \pmap(2))$ and $(\texttt{a}, \pmap(4))$.
Partition \texttt{\$} precedes partition \texttt{n}, so suffix 5 already sorts first, as in Figure~\ref{fig:prefix-doubling}.
Suffixes 1 and 3 stay tied, since $\pmap(2) = \pmap(4)$: their extension suffixes 2 and 4 share partition \texttt{n}, so their local ranks compare correctly.

\para{What target coherence buys}
A rank is only ever compared with ranks from its own partition, where local ranks order suffixes exactly as global ranks would, since the two differ by the partition's offset.
Ranks therefore need only enough bits to count the suffixes of one partition, 31 in \gpusaca, and key~\eqref{eq:keyt} fits in 62 bits, whatever the value of $n$.
Without coherence, a sender could still add its partition's offset and send global ranks, but these take $\lceil \log_2 n \rceil$ bits, 38 for the 238B tokens of \data{pile}, and key~\eqref{eq:keyt} would no longer fit a 64-bit word.

\subsubsection{\textbf{Rank messages}}
\label{sec:method-message}

In \gen{t \ge 1}, the only suffix that may need the rank of $m$ is $m - \cx{t - 1}$.
So a rank has at most one reader per generation, and its holder finds the reader by subtraction.
Once \gen{t} has ranked the active suffixes of a partition, each of them, $m \ge \cx{t}$, sends one \defn{message}
\begin{equation}
\label{eq:msg}
    \langle\, m - \cx{t}, \; \lrk{t}(m), \; \mathit{final} \,\rangle
\end{equation}
to its reader in \gen{t + 1}, where the bit \textit{final} is set if and only if the rank of $m$ is final, \ie $m$ retires in \gen{t} and its rank will not change in later generations.
We call a message with \textit{final} set a \defn{poison pill}, since its reader retires on receiving it, as shown below.
Messages are appended to the \defn{inbox} of the reader's partition, $\pmap(m - \cx{t})$.
Generations are separated by \defn{barriers}: barrier~$t$ precedes \gen{t}, and \gen{t + 1} reads the inboxes only after \gen{t} has written all of them.
Within a generation, the partitions are independent, and run in any order or concurrently.

\para{Receive-side join}
Before ranking a partition, \gpusaca\ sorts its active suffixes by index, and each message in its inbox finds its reader by binary search.
A hit supplies the reader's extension rank.
A miss means that the reader has already retired, and the message is discarded.
No suffix receives two messages in a generation (Theorem~\ref{thm:coverage}), so the messages join without conflicts.

\para{Forwarding final ranks}
A suffix stops sending once its rank is final.
Its rank may still be needed later, by a suffix whose extension suffix retired two or more generations earlier, and no message would then carry it.
Forwarding fills this gap.
Suppose that in \gen{t}, suffix $x$ receives the final rank of $m = x + \cx{t - 1}$.
In \gen{t + 1}, suffix $y = x - \cx{t - 1}$ needs the rank of $y + \cx{t} = m$, since $\cx{t} = 2\cx{t - 1}$. $x$ is the midpoint of $y$ and $m$.
So besides its own message, $x$ \defn{forwards} the rank that it received, as $\langle y, \lrk{}(m), \mathit{final} \rangle$ with \textit{final} set.
An active suffix whose extension suffix has already retired retires in the current generation. By Theorem~\ref{thm:tci}, the suffixes tied with it have their extension suffixes in the same partition, where the retired suffix's final rank is unique, so no suffix shares its key.
Hence $x$ retires in \gen{t}, its own message is a pill as well, and a suffix forwards at most once.
The same argument gives \defn{pullback}. If a suffix $j$ is active in \gen{t \ge 2}, then its extension suffix of \gen{t - 1}, $j + \cx{t - 2}$, was active in \gen{t - 1}.

\begin{theorem}[Complete coverage]
\label{thm:coverage}
In every \gen{t \ge 1}, each suffix receives at most one message, and each active suffix $j$ receives exactly one. It carries the rank of the extension suffix $m = j + \cx{t - 1}$, marked final if and only if $m$ retired before \gen{t}.
$m$ sends this message itself if it was active in \gen{t - 1}; otherwise, the midpoint $j + \cx{t - 2}$ forwards it.
\end{theorem}
\begin{proof}[Proof sketch]
By induction on $t$.
Only two suffixes can message $j$: $m$, with its own message, and the midpoint $z = j + \cx{t - 2}$, by forwarding.
If $m$ was active in \gen{t - 1}, $m$ sends its rank, and $z$ received that rank unmarked and forwards nothing.
Otherwise, $z$ was active in \gen{t - 1} by pullback, received the final rank of $m$ by induction, and forwards it to $j$.
Appendix~\ref{sec:supp-proof-coverage} gives the full proof.
\end{proof}

\para{Message volume and the liveness filter}
A suffix sends at most one message of its own and one forward in each generation that ranks it, and forwards at most once over a run.
\Gen{0} thus sends up to one message per suffix, and each later generation at most two per suffix that it ranks.
Messages to retired suffixes are \defn{dead}, and the join discards them.
A sender cannot know whether its reader will retire in the current generation, but it can know whether the reader already has.
As soon as a partition has been ranked, \gpusaca\ flags its retired suffixes in a bit beside their entries of \pmap, and sends no message to a flagged reader.
This \defn{liveness filter} removes only dead messages, so it changes neither the ranks nor Theorem~\ref{thm:coverage}.
It removes every message to a suffix that retired in an earlier generation, and within a generation, to one that retired in a partition processed before the sender's.

\subsection{Placing the streams}
\label{sec:method-residency}

With the rank table gone, nearly all of \gpusaca's data is streams, and placing the streams across the memory tiers needs no lookahead.

\para{Streams and random-access structures}
\gpusaca\ holds three kinds of data that grow with $n$.
The \defn{records} of each partition hold a suffix index and a local rank per suffix, laid out with the active suffixes first, so that each generation reads a prefix of them.
The messages fill the inboxes of each generation.
Records and messages are \defn{streams}: they are written and read sequentially, in chunks, and each chunk can reside in device memory, in host memory, or on SSD.
Between generations, \gpusaca\ keeps them in host memory or on SSD. Device memory holds only the few partitions in flight.
The third kind is the two structures read at random, the text $T$ and the partition map \pmap.
$T$ is last read by \gen{0}, to gather the $\ell$-contexts, while \pmap\ addresses every message and holds the liveness filter's flags until the last generation.
\gpusaca\ keeps both in host memory, and only one at a time when memory is short (\S\ref{sec:system-residency}).
In principle, neither has to be resident.
Partitioning could store each suffix's $\ell$-context in its record, and the messages could carry the partition addresses that the next generation needs, but on our corpora each costs more stream bytes than it saves in host memory (Appendices~\ref{sec:supp-contexts} and~\ref{sec:supp-relay}).

\para{Messages are elastic}
Messages can still make barrier~1's stream footprint large.
\Gen{0} may send one message per suffix, and each message includes a suffix index as well as a rank.
Unlike a rank table, however, messages are needed for only one generation and are read sequentially.
Their inboxes also shrink as suffixes retire, so the implementation can keep them in chunks across host memory and SSD.
As in the semi-external model~\cite{abello02}, what must be resident bounds $n$, and what may be resident bounds only the traffic.

\para{Lifetimes}
Almost every chunk is \defn{short-lived}. It crosses exactly one barrier, written before it and read right after it.
These are:
\begin{inparaenum}
    \item the suffix indices that partitioning scatters, read by \gen{0};
    \item the records of the active suffixes, which each generation writes for the next; and
    \item the messages, which each generation sends to the next.
\end{inparaenum}
The only exception is the records of retired suffixes, which wait in their partitions until output assembly reads them.

\para{Greedy placement}
All chunks share one host memory budget, whose free room only grows over a run, as $T$ is released.
\gpusaca\ places each chunk as it is written, in host memory if there is room.
Retired records have the lowest priority. When a short-lived chunk finds no room, it evicts retired records to SSD, just in time, and nothing else is ever evicted.
This greedy policy moves as few bytes to and from SSD as any placement can, up to the granularity of chunks.
A short-lived byte crosses one barrier, so any placement must spill the short-lived bytes that exceed the free room there.
The greedy policy, which never evicts them, spills no more.
A retired record is read only by output assembly, after every other use, so evicting it first costs nothing, as in Belady's rule~\cite{belady66} of evicting what is used farthest in the future.

\para{SSD traffic}
The SSD traffic of a run is therefore twice the short-lived bytes that exceed the free room at each barrier, summed over the barriers, plus twice the largest excess of retired records over the room that the short-lived bytes leave free, plus the output, written once.
With ample host memory, only the output reaches SSD.
Barrier~1 holds the most stream data on every corpus we index.
It carries every record, active or retired, and up to one message per suffix, together about 18 bytes per suffix (\S\ref{sec:system-layout}).
The liveness filter already removes about half of the dead messages of \gen{0} (\S\ref{sec:evaluation-structure}).

\begin{figure*}[t]
    \centering
    \begin{subfigure}[t]{\columnwidth}
        \centering
        \includegraphics{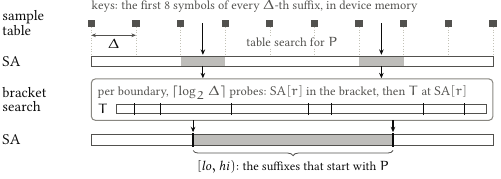}
        \caption{$P$ no longer than the keys}
        \label{fig:query-search-short}
    \end{subfigure}\hfill
    \begin{subfigure}[t]{\columnwidth}
        \centering
        \includegraphics{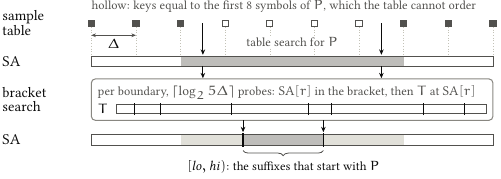}
        \caption{$P$ longer than the keys}
        \label{fig:query-search-long}
    \end{subfigure}
    \caption{
        Searching the index for a pattern $P$, top to bottom.
        The table search compares $P$ with only the keys of the sample table, and places each boundary of the interval $[\mathit{lo}, \mathit{hi})$ in a bracket (shaded, top): (a) between two consecutive samples, or (b) if keys equal the first 8 symbols of $P$ (hollow), which the table cannot order against $P$, in one bracket that spans their samples.
        The bracket search locates each boundary inside its bracket, reading \sa\ and the text $T$ from whichever tiers hold them, and yields the interval (shaded, bottom).
    }
    \label{fig:query-search}
\end{figure*}

\subsection{Queries over the index}
\label{sec:method-query}

A small sample of the suffix array, kept in device memory, narrows each search to one short bracket, which is then read from wherever the index lives.

\para{Finding a pattern}
The index consists of the text $T$ and its suffix array \sa.
A query for a pattern $P$ asks for the interval $[\mathit{lo}, \mathit{hi})$ of \sa\ whose suffixes start with $P$, whose length is the number of occurrences of $P$ in $T$~\cite{manber90}.
Two binary searches find the boundaries of the interval.
Each probe makes two dependent random reads: an entry of \sa, and then the text of that suffix, up to $|P|$ symbols.
On our corpora, a boundary takes 33--38 probes into a text and \sa\ of \qty{56}{\giga\byte} to \qty{1.7}{\tera\byte}.

\para{Two-step search}
A \defn{sample table}, also called a supra-index~\cite{ferragina08,baezayates96}, keeps the first 8 or 16 symbols of every \stride-th suffix of \sa, so its keys are sorted (Figure~\ref{fig:query-search}).
The \defn{table search} compares $P$ with these keys alone, and finds for each boundary a \defn{bracket} in which the boundary must lie.
Each \defn{bracket search} then binary-searches its bracket for its boundary, comparing against the text: at $\stride = 32$, in 5 probes into \sa, within 160 contiguous bytes.
A pattern longer than the keys ties with every sample whose key equals its first 8 or 16 symbols, since the table cannot order those samples against it.
Both of its boundaries then lie in one bracket that spans all the tied samples, and their searches take more probes (Figure~\ref{fig:query-search-long}).

\para{Placement}
A bracket search reads \sa\ only within one contiguous bracket, but reads the text at every probe, at scattered positions.
It therefore runs on the processor whose memory holds the text, the device or the host, and fetches each bracket from whichever tier holds \sa: device memory, host memory, or SSD.
The sample table stays in device memory: at $\stride = 32$ and 8-symbol keys, it is a fourteenth of the index, so it fits where the index does not (\S\ref{sec:system-query}).

\section{Streaming Partitions Through the GPU}
\label{sec:system}

\begin{figure*}[t]
    \centering
    \includegraphics{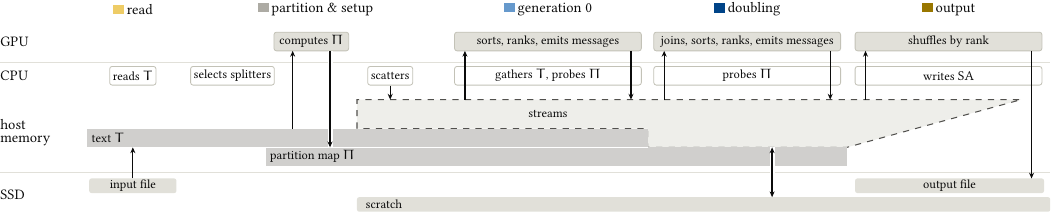}
    \caption{
        \gpusaca\ end to end: where each phase runs, and where its data lives.
        The columns show its five phases (\S\ref{sec:system-pipeline}).
        In the last three, each partition passes through the host's CPU cores to the GPU and back.
        $T$ and \pmap\ are held until the ends of \gen{0} and doubling.
        The streams (records and messages) are elastic: they take the rest of host memory's budget, and spill to the SSD beyond it (\S\ref{sec:system-residency}).
        If $T$ and \pmap\ do not fit together, \pmap\ stays on the SSD until $T$ is released (\S\ref{sec:system-residency}).
    }
    \label{fig:system-overview}
\end{figure*}

\gpusaca\ divides the work between the device and its host.
The device sorts and ranks the partitions, a few at a time, and the host prepares their data, looks up the text and the partition map, and moves every byte between host memory, SSD, and the device.
This section describes the pipeline that streams the partitions through the device (\S\ref{sec:system-pipeline}), where the streams live in between (\S\ref{sec:system-residency}), and the query engine (\S\ref{sec:system-query}).
Figure~\ref{fig:system-overview} shows where each phase runs and where its data lives.

\subsection{Pipeline}
\label{sec:system-pipeline}
\label{sec:system-device}

\para{Phases}
\gpusaca\ runs the construction in five phases (Figures~\ref{fig:system-overview} and~\ref{fig:perf-phases}).
It first reads $T$ into host memory.
During partitioning, the host selects splitters, the device computes the partition map and partition sizes, and the host scatters suffix indices into their partitions.
A background thread pins the buffers used for transfers.
\Gen{0} and the doubling generations then stream partitions through the device for sorting and ranking.
Finally, output assembly streams the records through the device, places each suffix at its final rank, and writes each partition to its output range.

\para{Thread pools}
Four thread pools, connected by queues, move each partition through a generation.
\defn{Loaders} read each partition's records and inbox from host memory or SSD into pinned host buffers called \defn{slabs}.
Each slab holds one chunk of records or messages (\S\ref{sec:system-residency}).
Because no suffix receives two messages in a generation (Theorem~\ref{thm:coverage}), an inbox is no larger than its partition.
A slab pool sized for the largest partition can therefore hold any partition's records and inbox.
\defn{Stagers} do the host's work on each slab, before and after the device's.
\defn{Dispatchers}, one per CUDA stream, issue the transfers and launch the kernels, and \defn{writers} write the partition's records and messages back to host memory or SSD.
Since the stagers and the dispatchers are separate, the host and the device work at the same time, on up to three partitions in flight, one per CUDA stream.

\para{Device work}
In a doubling generation, the device sorts a partition's active suffixes by index, joins its inbox to them, sorts their keys~\eqref{eq:keyt} in one 62-bit radix sort~\cite{merrill11}, assigns the ranks in one scan, separates the retired suffixes, and buckets the messages and forwards by destination partition.
\Gen{0} instead sorts by key~\eqref{eq:key0}, in two stable radix sorts.
A partition takes 97 bytes of device memory per suffix, mostly for the double buffers of the radix sorts.
\gpusaca\ sizes the partitions so that three of them, with their sort workspaces, fit in 90\% of the free device memory, and sets the target size $B$ so that partitions of up to $1.2B$ suffixes fit, which absorbs the imbalance of the sampling (\S\ref{sec:method-partition}).

\para{Host work}
Before a slab goes to the device, the stagers look up in \pmap\ the partitions that each suffix's message and possible forward go to, with the retirement flags of their readers, which the kernels use as the liveness filter (\S\ref{sec:method-message}).
At \gen{0}, they also gather each suffix's key~\eqref{eq:key0} from $T$ and \pmap.
After the device returns a partition, they flag its retired suffixes in \pmap.

\subsection{Streams in host memory and on SSD}
\label{sec:system-layout}
\label{sec:system-residency}

\para{Layout}
Suffix indices take 5 bytes and ranks 31 bits, and a message carries its \textit{final} bit in the top bit of its rank.
A record and a message thus both take 9 bytes, and the streams at barrier~1 take up to 18 bytes per suffix.
Each stream is stored in chunks of $2^{21}$ records or $2^{18}$ messages, which are the unit of transfer and of placement (\S\ref{sec:method-residency}).
The text, the spilled chunks, and the output bypass the page cache (\odirect), whose pages would otherwise compete for host memory, and each partition writes its range of the output at its offset in the file.

\para{Budget}
The host memory budget, 90\% of the available memory by default, holds $T$, \pmap, and the working buffers, chiefly the pinned slabs, and leaves the rest, the free room, to the chunks of the streams.
The free room only grows over a run: partitioning releases the scatter buffers, and \gen{0} releases $T$ after its last use.
On our server (\S\ref{sec:prelim-hw}), $T$ and \pmap\ fit the budget together up to about 270B tokens, and with \pmap\ on SSD, where only its flags stay in host memory, up to about 420B.

\para{Placement}
Each chunk is placed as it is written: in host memory if there is room, and on SSD otherwise.
The scatter fills all partitions at once, so it gives host memory to whole partitions, largest first.
From then on, records written back keep their place in host memory, messages fill host memory before they spill, and chunks of retired records take only free room.
A chunk of active records or messages that finds no room evicts a chunk of retired records to SSD, just in time, and takes its place.
This realizes the greedy policy of \S\ref{sec:method-residency}.

\para{The partition map on SSD}
When $T$, \pmap, and the working buffers do not fit the budget together, \gpusaca\ writes \pmap\ to SSD as partitioning computes it, and keeps only its retirement flags, one bit per suffix, in host memory until $T$ is released.
Meanwhile, \gen{0} recomputes the entries of \pmap\ that it needs, at $j + \ell$ and $j - \ell$ for each suffix $j$, by binary-searching the $K$-grams of $T$ at those positions among the splitters.
\pmap\ is then written once and read twice, and in exchange the free room at barriers~0 and~1 grows by nearly its size.

\subsection{Query engine}
\label{sec:system-query}

\begin{table}[t]
\caption{
    The placements of the index that our corpora land in: where the text and \sa\ live, where the bracket search runs, and the sample table of each corpus, its stride \stride, symbols per key, and size.
    The table search always runs on the device.
}
\label{tab:query-placements}
\small
\setlength{\tabcolsep}{3pt}
\begin{tabular}{@{}llllS[table-format=3.0]S[table-format=2.0]S[table-format=2.1]l@{}}
    \toprule
    & & & \multirow[b]{2}{*}{\makecell{Bracket\\search}} & \multicolumn{3}{c}{Sample table} & \\
    \cmidrule(lr){5-7}
    & Text & \sa & & {\stride} & {Key} & {GB} & Corpus \\
    \midrule
    A & device & device      & device & 32   & 8    & 4.0  & \data{dolma-sample} \\
    C & host   & host        & host   & 32   & 16   & 22.0 & \data{stackexchange} \\
      &        &             &        & 32   & 16   & 24.1 & \data{wikipedia} \\
      &        &             &        & 32   & 16   & 30.2 & \data{arxiv} \\
      &        &             &        & 32   & 16   & 63.9 & \data{github} \\
    D & host   & SSD         & host   & 128  & 16   & 49.5 & \data{c4} \\
      &        &             &        & 128  & 16   & 59.6 & \data{pile} \\
    \bottomrule
\end{tabular}
\end{table}

\para{Placement}
The index consists of the text and its 5-byte \sa.
At load time, the engine chooses the first placement in the order A, C, B, D that fits the available device and host memory (Table~\ref{tab:query-placements}).
In A, the device holds both the text and the suffix array, so it runs the bracket search. In C and D, the host runs that search.
B keeps the text in device memory but places the suffix array in pinned host memory. It follows C because the host can search its own memory faster than the device can fetch scattered suffix-array entries over PCIe (\S\ref{sec:prelim-hw}). None of our evaluated corpora uses B (Appendix~\ref{sec:supp-query-placements}).
The stride of the sample table starts at 32, and doubles until the table fits in the free device memory.
Its keys hold 8 symbols when the device runs the bracket search, and 16 when the host does, since the table then has the device memory to itself.

\para{Table search}
Batches of patterns stream through the device in chunks.
Each pattern's first 8 or 16 symbols are packed into an integer key, and two adjacent device threads search the table for its two boundaries; their searches probe the same keys until one equals the pattern's prefix, so their loads coalesce.
In A, each thread then binary-searches its bracket in device memory.

\para{Bracket search on the host}
In C and D, the host's cores search the brackets of one chunk while the device runs the table search on the next.
A binary search is a chain of dependent cache misses, so each core interleaves 32 searches, as in asynchronous memory access chaining~\cite{kocberber15}, and prefetches each probe's entry of \sa\ and suffix of the text a turn ahead.
In D, \sa\ stays memory-mapped on SSD, and before searching a chunk, the host requests the pages of the next chunk's brackets, so that the drives read them while the cores search.

\para{Checks}
The engine checks the order of \sa\ at load, and an \defn{interval check} verifies each interval without a search, in $\bigoh{|P|}$ time.

\section{Evaluation}
\label{sec:evaluation}

This section evaluates \gpusaca\ on seven corpora of 8 to 238B tokens (\S\ref{sec:evaluation-setup}).
We validate every suffix array that it builds, and test the method's invariants throughout its implementation (Appendix~\ref{sec:supp-validation}).
The section answers four questions.
\begin{compactenum}[Q1.]
    \item \emph{How fast does \gpusaca\ build the index?} (\S\ref{sec:evaluation-build})
    It builds the suffix array \numrange{3.0}{8.2}$\times$ faster than \ig's pipeline end to end, and its construction is \numrange{6.3}{11.4}$\times$ faster than that of \libsais, on every corpus that they build.
    On prefixes of \data{c4}, its lead over \ig\ grows with the length of the corpus.

    \item \emph{Does it scale past each memory tier?} (\S\ref{sec:evaluation-reach})
    Yes.
    Device memory bounds only the partition size. With a sixteenth of its default, about \qty{5}{\giga\byte}, builds take \numrange{1.3}{1.5}$\times$ as long.
    Dropping the rank table moves the wall of host memory on our server from about 130B to about 420B tokens.
    As host memory shrinks to its floor, builds slow by at most \num{3.3}$\times$, with no cliff.

    \item \emph{What determines its cost?} (\S\ref{sec:evaluation-structure}, \S\ref{sec:evaluation-cost})
    The repeat structure of a corpus, which varies the doubling work 30-fold across our corpora.
    Doubling costs about \qty{2}{\ns} per active suffix and generation, and spilling adds the SSDs' transfer time on \data{c4} and \data{pile}.

    \item \emph{Is the finished index fast to query?} (\S\ref{sec:evaluation-query})
    The batched query engine returns \ig's counts at median speedups of \numrange{10.8}{32.6}$\times$, with the index in device memory, in host memory, or on SSD.
\end{compactenum}

\subsection{Experimental setup}
\label{sec:evaluation-setup}

\subsubsection{\textbf{Platform}}
We run all experiments on the server of \S\ref{sec:prelim-hw}, whose measured rates Table~\ref{tab:hardware} lists.
\gpusaca\ is implemented in C++17 and CUDA, and is available at \url{https://github.com/saltsystemslab/Cascade}.
We built it with CUDA~12.8 and GCC~11.5, and it runs on Rocky Linux~9.7 (Linux~5.14) with NVIDIA driver 580.105.08.
We run it with $\ell = 16$-token initial contexts and $K = 8$-token splitters on every corpus, and $K = 20$ in one sweep (\S\ref{sec:evaluation-reach}); it sizes its partitions to the device memory (\S\ref{sec:system-device}).

\subsubsection{\textbf{Datasets}}
\label{sec:datasets}
We use seven tokenized corpora that \ig\ indexes~\cite{liu24} (Table~\ref{tab:datasets}).
\data{dolma-sample} is an 8B-token sample of the Dolma~v1.6 pretraining corpus~\cite{soldaini24}, \data{c4} the English split of C4~\cite{raffel20}, cleaned web text, and \data{pile} the deduplicated Pile~\cite{gao20}, a 22-domain mixture on which the Pythia models were trained~\cite{biderman23}.
Four are slices of Red\-Pajama-V1~\cite{weber24}: the question-answer threads of the 28 largest Stack Exchange sites (\data{stackexchange}), Wikipedia articles in 20 languages (\data{wikipedia}), the \LaTeX\ sources of preprints (\data{arxiv}), and permissively licensed source files (\data{github}).
Their domains differ in repetitiveness, which sets the work of prefix doubling (\S\ref{sec:prelim-doubling}); deduplication does not remove long repeats, as \data{pile} still repeats a span of 2.2M tokens.
Each corpus is an array of 16-bit token IDs, with a separator token before each document: \ig's released array for \data{dolma-sample}, and the others tokenized with the Llama-2 tokenizer~\cite{touvron23}, as \ig\ does.

\begin{table}[!htb]
\caption{
    Evaluation corpora and their structure (\S\ref{sec:evaluation-structure}): the number of 16-bit tokens, in billions, and the longest repeat, in tokens;
    the share of suffixes still active after generation~0, and the mean number $\bar{d}$ of doubling generations per suffix;
    and the messages per suffix that a whole run sends through the liveness filter.
}
\label{tab:datasets}
\small
\setlength{\tabcolsep}{4pt}
\begin{tabular}{@{}l S[table-format=3.2] S[table-format=7.0] S[table-format=2.1, table-space-text-post=\%] S[table-format=1.2]
                   S[table-format=1.2]@{}}
    \toprule
    & {Tokens} & {Longest} & \multicolumn{2}{c}{Doubling} & {Messages} \\
    \cmidrule(lr){4-5}
    Corpus & {(B)} & {repeat} & {active} & {$\bar{d}$} & {per suffix} \\
    \midrule
    \data{dolma-sample}  & 8.04   & 171633  & 7.1\%  & 0.15 & 0.72 \\
    \data{stackexchange} & 21.97  & 99999   & 22.3\% & 0.53 & 1.01 \\
    \data{wikipedia}     & 24.12  & 65160   & 13.7\% & 0.25 & 0.75 \\
    \data{arxiv}         & 30.21  & 520234  & 18.0\% & 0.61 & 1.09 \\
    \data{github}        & 63.94  & 518745  & 77.2\% & 4.50 & 4.78 \\
    \data{c4}            & 198.08 & 88956   & 12.3\% & 0.22 & 0.69 \\
    \data{pile}          & 238.43 & 2197906 & 20.9\% & 0.64 & 1.08 \\
    \bottomrule
\end{tabular}
\end{table}

\subsubsection{\textbf{Baselines}}
We compare against \ig's indexing pipeline~\cite{liu24}, which is built for this workload, and against the state of the art in each family of \sa\ constructors (\S\ref{sec:introduction}).
\ig's pipeline (v2.6.0~\cite{infinigram}) also represents the external-memory CPU constructors.
It adopts \citet{lee22}'s parallel block-and-merge scheme, suffix-sorting overlapping blocks of the text in memory with SA-IS~\cite{nong11,ko03}, and merging their suffix arrays from SSD.
We build each \ig\ index as one shard on all 128 cores, with the same memory budget as \gpusaca's, 90\% of the available memory (\S\ref{sec:system-residency}).
Its merge compares suffixes on at most their first \num{50000} tokens, so it does not guarantee an \emph{exact} suffix array when longer repeats exist, as they do in all our corpora (Table~\ref{tab:datasets}); \gpusaca's is exact.
\libsais~\cite{libsais} (v2.10.4), a multithreaded SA-IS~\cite{nong11,karkkainen09,timoshevskaya14,xie20} library, represents the in-memory CPU constructors.
We run its 16-bit-alphabet, 64-bit-index variant in the fastest configuration we found, 64 threads on one socket (32 perform alike, 128 slower), and unbound on \data{github}, whose \qty{0.64}{\tera\byte} footprint exceeds one socket's memory.
\libcubwt~\cite{libcubwt,osipov12,wang16,buren19} (v1.0.0, its only release that constructs suffix arrays) represents the GPU constructors; it supports only byte alphabets, and device memory bounds it below our smallest corpus.
Of the standalone external-memory constructors, pSAscan~\cite{karkkainen15} supports only byte alphabets, and \sufrank~\cite{barsky20} completed on none of our corpora.

\para{Measurements}
Every reported run had exclusive use of the server.
We time every build end to end, from reading the token array until the index is durably written, except against \libsais\ in \S\ref{sec:evaluation-build}, where we time the construction alone.
We report wall-clock time, and throughput as tokens indexed per second, as the median of three runs per instance, and the sweeps of \S\ref{sec:evaluation-reach} as means; \S\ref{sec:evaluation-query} describes how we time queries.

\subsection{Build time}
\label{sec:evaluation-build}

Table~\ref{tab:baseline} compares the build times of \gpusaca\ with those of \ig's pipeline and \libsais\ (\S\ref{sec:evaluation-setup}); \libcubwt\ builds none of the corpora.

\begin{table}[!htb]
\caption{
    Build times in minutes, end to end against \ig's pipeline and of the construction alone against \libsais, which runs out of memory (OOM) on the two largest corpora, and \gpusaca's speedups ($\times$), computed from the times in seconds.
}
\label{tab:baseline}
\small
\setlength{\tabcolsep}{3pt}
\begin{tabular}{@{}l S[table-format=3.0] S[table-format=2.0] S[table-format=1.1]
                   S[table-format=2.0] S[table-format=2.0] S[table-format=2.1]@{}}
    \toprule
    & \multicolumn{3}{c}{End to end} & \multicolumn{3}{c}{Construction} \\
    \cmidrule(lr){2-4} \cmidrule(l){5-7}
    Corpus & {\ig} & {\gpusaca} & {$\times$} & {\libsais} & {\gpusaca} & {$\times$} \\
    \midrule
    \data{dolma-sample}  & 3    & 1  & 3.0  & 8     & 1  & 10.5 \\
    \data{stackexchange} & 10   & 2  & 4.7  & 15    & 2  & 8.4 \\
    \data{wikipedia}     & 12   & 2  & 5.4  & 19    & 2  & 11.4 \\
    \data{arxiv}         & 13   & 3  & 4.4  & 17    & 2  & 7.3 \\
    \data{github}        & 113  & 14 & 8.2  & 79    & 12 & 6.3 \\
    \data{c4}            & 104  & 23 & 4.6  & {OOM} & 19 & {--} \\
    \data{pile}          & 134  & 37 & 3.6  & {OOM} & 31 & {--} \\
    \bottomrule
\end{tabular}
\end{table}

\para{Against \ig}
Both pipelines read the same token array, and write its suffix array to SSD, at five bytes per suffix.
\gpusaca\ builds the index \num{3.0}--\num{8.2}$\times$ faster than \ig, and builds it exactly, whereas \ig\ truncates its suffix comparisons (\S\ref{sec:evaluation-setup}).
Across the corpora, \ig\ indexes 9.5--49M tokens/s, and \gpusaca\ 78--184M tokens/s.
\ig's merge compares suffixes byte by byte, at a cost that grows with their common prefix up to its \num{100000}-byte bound, whereas the number of \gpusaca's doubling generations that a suffix takes part in grows only with the logarithm of its repeat depth (\S\ref{sec:prelim-doubling}).
\gpusaca's lead is thus smallest on \data{dolma-sample}, the least repetitive corpus (\num{3.0}$\times$), and largest on \data{github}, the most repetitive (\num{8.2}$\times$), on which \ig\ runs 3--5$\times$ slower per token than on the other corpora.
On \data{c4}, \gpusaca\ keeps nearly its in-memory rate, 144M tokens/s against 147--184M on the four smaller corpora, and leads by \num{4.6}$\times$.
On \data{pile}, its lead is the second smallest (\num{3.6}$\times$), as its rate falls to 107M tokens/s: its generation~0 spills about \qty{3.4}{\tera\byte} to SSD, which raises that phase's cost from 2.4--2.5~ns per token in memory to 4.4~ns, against 3.1~ns on \data{c4} (\S\ref{sec:evaluation-cost}).

\para{Against \libsais}
On the construction alone, \gpusaca\ is \num{6.3}--\num{11.4}$\times$ faster than \libsais.
\libsais\ indexes the corpora within its wall (\S\ref{sec:introduction}) at 14--29M tokens/s, whereas \gpusaca's throughput on that subset ranges 85--242M tokens/s.
\gpusaca's lead is smallest on \data{github}, on which it does 30$\times$ as much doubling work per token as on \data{dolma-sample} (\S\ref{sec:evaluation-structure}), whereas SA-IS does linear work regardless of the repeats.
\libsais\ runs out of host memory on \data{c4} and \data{pile}.
At its peak footprint of \num{10.1}--\num{10.4} bytes per token, host memory bounds it to about 100B tokens, whereas \gpusaca's construction alone takes 19 and 31 minutes on them.

\para{Scaling with corpus length}
Figure~\ref{fig:perf-reach} grows a prefix of \data{c4} from 8B tokens to the whole corpus geometrically, each cut at a document boundary, and times each method end to end.
\gpusaca's throughput rises from 126 to 173M tokens/s up to 32B tokens, as it amortizes its fixed costs, and settles at 144--156M tokens/s from 64B on, as the prefixes spill \qty{0.19}{}, \qty{1.5}{}, and \qty{3.3}{\tera\byte} to SSD.
\ig\ holds 56M tokens/s up to 16B tokens and slows to 30--32M tokens/s from 64B on, so \gpusaca's lead grows from \num{2.3}$\times$ at 8B to \num{4.8}$\times$ at 64B, and stays at \numrange{4.6}{4.7}$\times$ beyond.
\libsais\ holds 16--18M tokens/s up to 32B tokens, slows to 6.4M at 64B, where its footprint spans both sockets' memory, and runs out of memory at 128B.

\begin{figure}[t]
    \centering
    \includegraphics{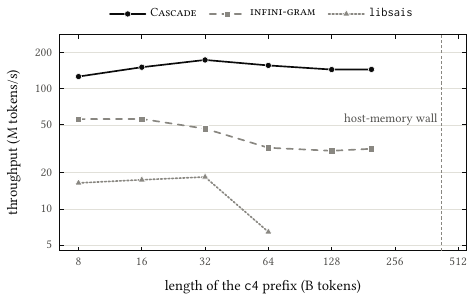}
    \caption{
        End-to-end throughput of each method on prefixes of \data{c4}, against their length; the last point is the whole corpus, and \libsais\ runs out of host memory from 128B tokens on.
        The dashed line marks the length at which the host budget fills, the wall for \gpusaca\ and \ig\ on our server.
    }
    \label{fig:perf-reach}
\end{figure}

\subsection{Reach across the memory hierarchy}
\label{sec:evaluation-reach}

Each of our first three contributions removes or softens one memory limit, and we test each in turn.
Raising the partition count emulates smaller devices (contribution 1).
Lowering the host budget forces the streams to spill across the tiers (contribution 3).
Projecting the host-memory wall with and without a rank table isolates the rank exchange (contribution 2).

\para{Device memory}
Figure~\ref{fig:perf-device} emulates smaller devices by raising the partition counts of \data{github} and \data{arxiv} from their device floors, 248 and 118, to 16 times as many, which halves the device's peak memory with each doubling, from about \qty{79}{} to \qty{5}{\giga\byte}.
From 8$\times$ on, \data{github}'s heaviest 8-gram outweighs a partition; the figure shows \data{github} at $K = 20$, which splits it at all five counts, building within 3\% of the times at $K = 8$.
The builds stay mostly flat up to 4$\times$.
Past 4$\times$, \gen{0} and doubling sort and rank up to 16 times as many partitions.
With a sixteenth of the default device memory, about \qty{5}{\giga\byte}, the builds take \num{1.3}$\times$ and \num{1.5}$\times$ as long on \data{arxiv} and \data{github} (Appendix~\ref{sec:supp-device}).
The device thus bounds the partition size and not the corpus, as long as it holds the partition of the heaviest $K$-gram, which longer splitters shrink.

\begin{figure}[t]
    \centering
    \includegraphics{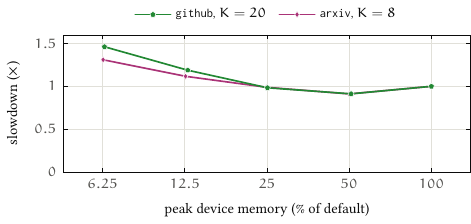}
    \caption{
        The end-to-end build's slowdown on \data{github} and \data{arxiv} against the device's measured peak memory, as their partition counts double from the device floors, $p = 248$ and $118$, to 16$\times$; each doubling halves the peak.
        Points are the mean of three runs, each coordinate relative to the default run's at the device floor, whose peak is about \qty{79}{\giga\byte}.
    }
    \label{fig:perf-device}
\end{figure}

\para{Host budget}
Figure~\ref{fig:perf-budget} lowers the host budget from its default to the floor: on \data{github}, which sends the most messages per token and keeps the most suffixes active (Table~\ref{tab:datasets}), and on \data{arxiv}, which spills nothing at the default.
Neither curve shows a cliff.
\data{github}'s spill grows from \qty{0.4}{} to \qty{5.4}{\tera\byte}, and its build slows down by up to \num{2.9}$\times$.
\data{arxiv} first spills at \qty{537}{\giga\byte}, below its \qty{584}{\giga\byte} peak footprint, and slows down by up to \num{1.4}$\times$.
Below the budgets at which the text and the partition map fit together, \qty{260}{} and \qty{150}{\giga\byte}, the map stays on SSD through generation~0.
That costs its size in writes and twice in reads, \qty{72}{} and \qty{30}{\giga\byte}, with no significant slowdown across the switch.
\gpusaca\ refuses budgets below the absolute floor, \qty{189}{} and \qty{117}{\giga\byte}, which the text, its per-suffix retirement bits, and \qty{53}{\giga\byte} of slabs and scatter buffers take up.
The floor forces every stream and the map to SSD; there, the builds slow by \num{3.3}$\times$ and \num{1.7}$\times$.

\begin{figure}[t]
    \centering
    \includegraphics{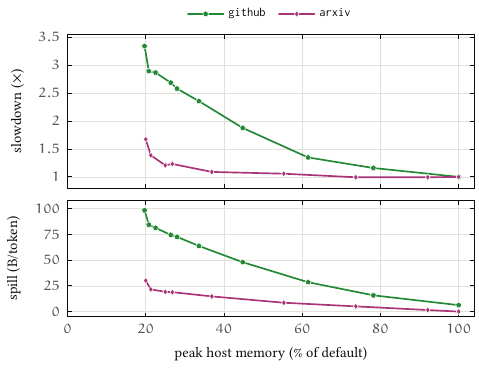}
    \caption{
        The end-to-end build's slowdown against its time at the default budget (top), and its spill in bytes per token additional to writing the output (bottom), as the host budget falls to each corpus' floor.
        Points, the mean of two or three runs, are their peak host memory, as a percentage of the peak at the default budget: \qty{963}{} and \qty{584}{\giga\byte} on \data{github} and \data{arxiv}.
        Each corpus' leftmost point is its floor, with every stream and the map forced to SSD.
    }
    \label{fig:perf-budget}
\end{figure}

\para{Walls}
Host memory bounds \gpusaca's reach where its floor, the text with its flags and buffers, fills the budget: at about 420B tokens on our server, twice the length of \data{c4} (\S\ref{sec:system-residency}).
A resident rank table of 40-bit entries would add 5 bytes per token to that, and wall the same design at about 130B tokens, below \data{c4} and \data{pile}.
Dropping the rank table lets \gpusaca\ build both.

\subsection{What the corpus costs}
\label{sec:evaluation-structure}

\para{Doubling work}
The repeat depths of a corpus fix its active sets.
A suffix stays tied after \gen{t} only while another suffix shares its $\cx{t}$-context, \ie while $\rdep(j) \ge \cx{t}$.
Figure~\ref{fig:structure-decay}a overlays the measured active sets on the repeat-depth distribution.
They agree to within 8\%, as \gpusaca\ retires some suffixes early (Appendix~\ref{sec:supp-structure}).
The corpora differ widely.
\Gen{0} leaves 7\% of the suffixes of \data{dolma-sample} active but 77\% of those of \data{github}, and thereafter about half of the active set of \data{dolma-sample} retires in each generation, whereas \data{github} retains 84--89\% of its own through \gen{5}.
Summed over the doubling generations, the active suffixes (\S\ref{sec:method-doubling}) make up the doubling work, and divided by $n$, the mean number $\bar{d}$ of doubling generations per suffix: 0.15 in \data{dolma-sample} and 4.5 in \data{github}, which does 30 times as much doubling work per token (Table~\ref{tab:datasets}).

\begin{figure}[t]
    \centering
    \includegraphics{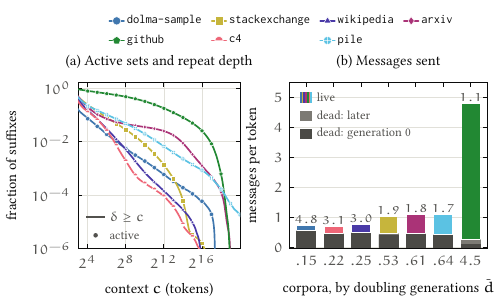}
    \caption{
        (a) The share of suffixes still active after each generation (markers), against the share of suffixes whose repeat depth is at least $c$ (lines).
        (b) The messages per suffix that a run sends, with the corpora in order of their mean doubling generations $\bar{d}$.
        The live messages, one per active suffix per generation, sum to $\bar{d}$ (colored); the dead ones that the liveness filter lets through are mostly \gen{0}'s (dark).
        Atop each bar: the messages that the run sends per live message.
    }
    \label{fig:structure-decay}
\end{figure}

\para{Message volume}
A run's message volume grows with both its doubling work and the size of the corpus.
Without the liveness filter (\S\ref{sec:method-message}), a run would send 1.21--6.18 messages per token; the filter skips \numrange{23}{48}\% of them, so a run sends from 0.69 messages per token on \data{c4} to 1.09 on \data{arxiv}, and 4.78 on \data{github} (Table~\ref{tab:datasets}).
Besides the live messages, 0.15--4.5 per token, a run sends 0.28--0.57 dead ones (Figure~\ref{fig:structure-decay}b).
90--97\% of the dead messages of a run are \gen{0}'s, but 43\% on \data{github}, where most suffixes survive \gen{0}.
The filter skips 85--90\% of the dead messages after \gen{0}, but only 41--49\% of its \gen{0} flood (Appendix~\ref{sec:supp-structure}).
The messages sent per live message thus fall as the doubling work grows.
\data{dolma-sample} discards 79\% of the messages it sends, and \data{github} only 6\%.

\para{Partition balance}
At the default partition counts, the largest partition exceeds its target size $B$ by \numrange{2.9}{5.3}\%, within the \num{20}\% slack that \gpusaca\ reserves (\S\ref{sec:system-device}), and the heaviest 8-grams hold only \numrange{0.1}{21}\% of $B$ (\S\ref{sec:method-partition}).
Sweeping $p$ up to \num{24522} and $K$ from 4 to 32 tokens keeps the largest partition within \num{10}\% of $B$, as the sampling predicts; only a $K$-gram that outweighs $B$, as \data{github}'s at $K = 4$, makes a run infeasible (Appendix~\ref{sec:supp-structure}).

\subsection{Where the bytes and the time go}
\label{sec:evaluation-cost}

\para{Bytes}
Figure~\ref{fig:perf-bytes} accounts for the bytes that each phase reads and writes.
The four smaller corpora (Table~\ref{tab:datasets}) read only their 2-byte tokens and write only the 5-byte suffix array.
The larger three also spill 6.3--19.6 bytes per token, up to \qty{3.5}{\tera\byte} at a time.
Partitioning and generation~0 write almost all of it, and the later phases read each byte back once.
The reads match the SSD traffic that \S\ref{sec:method-residency} predicts to within 2\%, and the gap is due to the page cache.

\begin{figure}[t]
    \centering
    \includegraphics{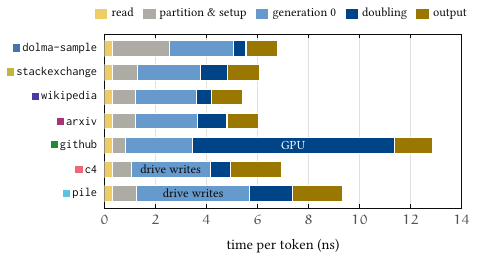}
    \caption{
        Time per token of \gpusaca's end-to-end build by phase.
        The phases are consecutive; partition \& setup cover partitioning and the overlapped pinning of the slabs.
        Labels name the resource that binds a phase taking about half the build.
    }
    \label{fig:perf-phases}
\end{figure}

\para{Seconds}
Figure~\ref{fig:perf-phases} splits each build into its phases (\S\ref{sec:system-pipeline}).
On the four corpora whose streams fit in host memory, every phase but doubling costs nearly the same per token.
Reading the tokens costs \qty{0.31}{\ns}, partitioning \qtyrange{0.5}{0.7}{\ns}, generation~0 \qtyrange{2.4}{2.5}{\ns}, and output assembly \qty{1.2}{\ns}.
Doubling instead follows the work: it costs \qtyrange{1.8}{2.2}{\ns} per active suffix and generation on \data{stackexchange}, \data{wikipedia}, \data{arxiv}, and \data{github}, though their $\bar{d}$ differs 18-fold, and dominates only \data{github}'s build, at \qty{7.9}{\ns} per token, three fifths of it, bound by the device.
Spilling adds the time of its drive writes to the phases that write.
On \data{c4} and \data{pile}, partitioning's scatter writes the spilled suffixes, and generation~0 costs \qty{3.1}{} and \qty{4.4}{\ns} per token, as it writes \qty{2.6}{} and \qty{3.4}{\tera\byte} of spill at \qty{4.2}{} and \qty{3.2}{\giga\byte\per\second} across the two drives.
The smallest build, \data{dolma-sample}'s, also waits \qty{12}{\second}, nearly a quarter of it, for its staging slabs to pin (Appendix~\ref{sec:supp-phases}).

\subsection{Query throughput}
\label{sec:evaluation-query}

\begin{figure}[t]
    \centering
    \includegraphics{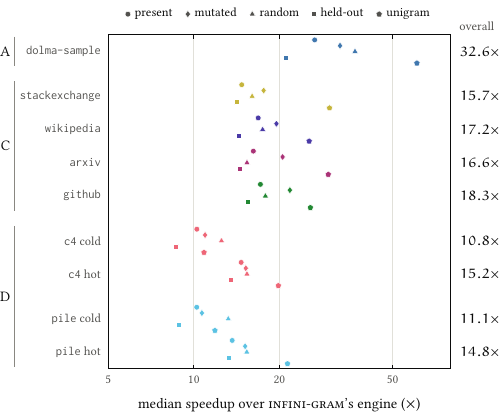}
    \caption{
        \gpusaca's query engine's median speedup over \ig's, per workload type, end to end, per corpus in its placement (Table~\ref{tab:query-placements}); with the suffix array on SSD (D), with cold and hot caches.
        At the right, the median over all 18 workloads.
    }
    \label{fig:query-speedup}
\end{figure}

\para{Setup}
We run both engines on all seven corpora, each in the placement that \gpusaca's engine picks (Table~\ref{tab:query-placements}): \data{dolma-sample} in A, the next four in C, and \data{c4} and \data{pile}, whose indexes exceed host memory, in D, where \ig\ memory-maps its \sa\ on SSD.
We run \ig's queries in parallel, one per core, without prefetching, as its API server does~\cite{infinigram}; its prefetching slows it on all seven corpora.

\para{Workloads}
Each corpus runs the same 18 workloads, seeded batches of five types that span hit rates (Appendix~\ref{sec:supp-query-workloads}): \emph{present} patterns, substrings at random positions; \emph{mutated} ones, with one token changed; \emph{random} ones, of random tokens; \emph{held-out} ones, substrings of another corpus; and all \num{65536} single tokens, the \emph{unigrams}.
Their lengths follow \ig's evaluation~\cite{liu24}: a batch holds \num{e6} patterns of 1 to 10 tokens, or \num{e5} of 100 or 1000.

\para{Measurement}
We time each engine end to end, from a batch in host memory to its counts.
Per corpus, the two engines alternate three runs each; a run loads the index (untimed; Appendix~\ref{sec:supp-query-measurement}) and repeats every workload six times (three in D), and we report the median of the three runs' medians.
In D, each engine drops its \sa\ from the page cache before every cold repetition, and keeps it in the hot ones.
Our counts match \ig's on every batch, and every interval passes its check (\S\ref{sec:system-query}).

\para{In memory}
\gpusaca's query engine is faster on every workload when the index fits in memory (Figure~\ref{fig:query-speedup}, Appendix~\ref{sec:supp-query-results}), by medians of \num{32.6}$\times$ in placement A and \numrange{15.7}{18.3}$\times$ in C.
In A, the device runs both searches; in C, the host's cores search only the brackets, in 5 probes per boundary for patterns no longer than the keys, and interleave 32 searches each.
Where we also ran a corpus in its other placements, the one that our engine picks is the fastest (Appendix~\ref{sec:supp-query-placements}).

\para{On SSD}
In placement D, our engine is \num{10.8}$\times$ and \num{11.1}$\times$ faster on \data{c4} and \data{pile} cold, and \num{15.2}$\times$ and \num{14.8}$\times$ hot.
In the cold setting, both engines wait on the drives.
Per pattern, \gpusaca's engine reads \numrange{0.03}{1.4} pages and \ig's \numrange{0.28}{11.6}, as \gpusaca\ reads only the brackets' pages, and requests them ahead of the search (\S\ref{sec:system-query}).

\section{Conclusion}
\label{sec:conclusion}

\gpusaca\ brings GPU acceleration to out-of-core suffix-array construction for LLM corpora.
Output-ordered partitions avoid merges, streamed rank exchange replaces the rank table with partition-local ranks, and the remaining streams spill across GPU memory, host memory, and SSD.
Every new corpus or tokenizer needs its own index, and on our server, \gpusaca\ builds each \numrange{3.0}{8.2}$\times$ faster than \ig's pipeline.
It indexes corpora much larger than GPU memory, such as \data{pile}, whose text alone is about five times the GPU's \qty{96}{\giga\byte}.
\gpusaca\ slows without a cliff as memory shrinks, by at most \num{1.5}$\times$ with only \qty{5}{\giga\byte} of GPU memory, a sixteenth of its default, and \num{3.3}$\times$ at the floor of host memory.
Its batched query engine reproduces \ig's counts at median speedups of \numrange{10.8}{32.6}$\times$.

More broadly, GPU joins, graph systems, and vector search also stop scaling once one complete structure no longer fits in memory.
The insight that may carry over is that a shared structure whose values each have one or a few readers known in advance, like the rank table, is a candidate for streaming to those readers.

Since partitions are independent within a generation, several GPUs or nodes are natural next steps, toward one index over a trillion-token corpus instead of shards.
So is building compressed indexes, such as the FM-index~\cite{xu25} and the suffixient array~\cite{cenzato26}, from \gpusaca's suffix arrays.
We expect \gpusaca\ to speed up both the preprocessing of LLM pretraining corpora, such as deduplication~\cite{wang26,lee22,kandpal22}, and the systems that query them: training data attribution~\cite{liu25}, non-parametric and hybrid language models~\cite{liu24}, and retrieval-based speculative decoding~\cite{he24,ho24}.
 
\bibliographystyle{ACM-Reference-Format}
\bibliography{refs}

%%% -*-BibTeX-*-
%%% Do NOT edit. File created by BibTeX with style
%%% ACM-Reference-Format-Journals [18-Jan-2012].

\begin{thebibliography}{56}

%%% ====================================================================
%%% NOTE TO THE USER: you can override these defaults by providing
%%% customized versions of any of these macros before the \bibliography
%%% command.  Each of them MUST provide its own final punctuation,
%%% except for \shownote{} and \showURL{}.  The latter does
%%% not use final punctuation, in order to avoid confusing it with
%%% the Web address.
%%%
%%% To suppress output of a particular field, define its macro to expand
%%% to an empty string, or better, \unskip, like this:
%%%
%%% \newcommand{\showURL}[1]{\unskip}   % LaTeX syntax
%%%
%%% \def \showURL #1{\unskip}           % plain TeX syntax
%%%
%%% ====================================================================

\ifx \showCODEN    \undefined \def \showCODEN     #1{\unskip}     \fi
\ifx \showISBNx    \undefined \def \showISBNx     #1{\unskip}     \fi
\ifx \showISBNxiii \undefined \def \showISBNxiii  #1{\unskip}     \fi
\ifx \showISSN     \undefined \def \showISSN      #1{\unskip}     \fi
\ifx \showLCCN     \undefined \def \showLCCN      #1{\unskip}     \fi
\ifx \shownote     \undefined \def \shownote      #1{#1}          \fi
\ifx \showarticletitle \undefined \def \showarticletitle #1{#1}   \fi
\ifx \showURL      \undefined \def \showURL       {\relax}        \fi
% The following commands are used for tagged output and should be
% invisible to TeX
\providecommand\bibfield[2]{#2}
\providecommand\bibinfo[2]{#2}
\providecommand\natexlab[1]{#1}
\providecommand\showeprint[2][]{arXiv:#2}
\makeatletter
\@ifundefined{NAT@parse@date}{}{\let\NAT@parse@date@orig\NAT@parse@date}
\@ifundefined{NAT@parse@date}{}{\def\NAT@parse@date#1#2#3#4#5#6@@{\NAT@parse@date@orig#1#2#3#4#5#6@@\def\NAT@tempyear{0000}\def\NAT@tempexlab{{?}}\ifx\NAT@year\NAT@tempyear\ifx\NAT@exlab\NAT@tempexlab\def\NAT@date{[n.\,d.]}\else\edef\NAT@date{[n.\,d.]\NAT@exlab}\fi\fi}}
\makeatother

\bibitem[Abello et~al\mbox{.}(2002)]%
        {abello02}
\bibfield{author}{\bibinfo{person}{James Abello}, \bibinfo{person}{Adam~L.
  Buchsbaum}, {and} \bibinfo{person}{Jeffery~R. Westbrook}.}
  \bibinfo{year}{2002}\natexlab{}.
\newblock \showarticletitle{A Functional Approach to External Graph
  Algorithms}.
\newblock \bibinfo{journal}{\emph{Algorithmica}} \bibinfo{volume}{32},
  \bibinfo{number}{3} (\bibinfo{year}{2002}), \bibinfo{pages}{437--458}.
\newblock
\href{https://doi.org/10.1007/s00453-001-0088-5}{doi:\nolinkurl{10.1007/s00453-001-0088-5}}


\bibitem[Aggarwal and Vitter(1988)]%
        {aggarwal88}
\bibfield{author}{\bibinfo{person}{Alok Aggarwal} {and}
  \bibinfo{person}{Jeffrey~S. Vitter}.} \bibinfo{year}{1988}\natexlab{}.
\newblock \showarticletitle{The input/output complexity of sorting and related
  problems}.
\newblock \bibinfo{journal}{\emph{Commun. ACM}} \bibinfo{volume}{31},
  \bibinfo{number}{9} (\bibinfo{year}{1988}), \bibinfo{pages}{1116--1127}.
\newblock
\href{https://doi.org/10.1145/48529.48535}{doi:\nolinkurl{10.1145/48529.48535}}


\bibitem[Baeza-Yates et~al\mbox{.}(1994)]%
        {baezayates96}
\bibfield{author}{\bibinfo{person}{Ricardo Baeza-Yates},
  \bibinfo{person}{Eduardo~Fernandes Barbosa}, {and} \bibinfo{person}{Nivio
  Ziviani}.} \bibinfo{year}{1994}\natexlab{}.
\newblock \showarticletitle{Hierarchies of indices for text searching}. In
  \bibinfo{booktitle}{\emph{Intelligent Multimedia Information Retrieval
  Systems and Management - Volume 1}} (New York, New York)
  \emph{(\bibinfo{series}{RIAO '94})}. \bibinfo{publisher}{LE CENTRE DE HAUTES
  ETUDES INTERNATIONALES D'INFORMATIQUE DOCUMENTAIRE}, \bibinfo{address}{Paris,
  FRA}, \bibinfo{pages}{295--312}.
\newblock


\bibitem[Barsky et~al\mbox{.}(2020)]%
        {barsky20}
\bibfield{author}{\bibinfo{person}{Marina Barsky}, \bibinfo{person}{Jonathan
  Gabor}, \bibinfo{person}{Mariano~P. Consens}, {and} \bibinfo{person}{Alex
  Thomo}.} \bibinfo{year}{2020}\natexlab{}.
\newblock \showarticletitle{Suffix rank: a new scalable algorithm for indexing
  large string collections}.
\newblock \bibinfo{journal}{\emph{Proc. VLDB Endow.}} \bibinfo{volume}{13},
  \bibinfo{number}{12} (\bibinfo{date}{July} \bibinfo{year}{2020}),
  \bibinfo{pages}{2787--2800}.
\newblock
\showISSN{2150-8097}
\href{https://doi.org/10.14778/3407790.3407861}{doi:\nolinkurl{10.14778/3407790.3407861}}


\bibitem[Beamer et~al\mbox{.}(2012)]%
        {beamer12}
\bibfield{author}{\bibinfo{person}{Scott Beamer}, \bibinfo{person}{Krste
  Asanovic}, {and} \bibinfo{person}{David Patterson}.}
  \bibinfo{year}{2012}\natexlab{}.
\newblock \showarticletitle{Direction-optimizing Breadth-First Search}. In
  \bibinfo{booktitle}{\emph{Proceedings of the 2012 International Conference
  for High Performance Computing, Networking, Storage and Analysis}}
  \emph{(\bibinfo{series}{SC '12})}. \bibinfo{publisher}{IEEE Computer
  Society}, \bibinfo{address}{USA}, \bibinfo{pages}{1--10}.
\newblock
\showISBNx{9781467308069}
\href{https://doi.org/10.1109/SC.2012.50}{doi:\nolinkurl{10.1109/SC.2012.50}}


\bibitem[Belady(1966)]%
        {belady66}
\bibfield{author}{\bibinfo{person}{L.~A. Belady}.}
  \bibinfo{year}{1966}\natexlab{}.
\newblock \showarticletitle{A study of replacement algorithms for a
  virtual-storage computer}.
\newblock \bibinfo{journal}{\emph{IBM Systems Journal}} \bibinfo{volume}{5},
  \bibinfo{number}{2} (\bibinfo{year}{1966}), \bibinfo{pages}{78--101}.
\newblock
\href{https://doi.org/10.1147/sj.52.0078}{doi:\nolinkurl{10.1147/sj.52.0078}}


\bibitem[Biderman et~al\mbox{.}(2023)]%
        {biderman23}
\bibfield{author}{\bibinfo{person}{Stella Biderman}, \bibinfo{person}{Hailey
  Schoelkopf}, \bibinfo{person}{Quentin Anthony}, \bibinfo{person}{Herbie
  Bradley}, \bibinfo{person}{Kyle O'Brien}, \bibinfo{person}{Eric Hallahan},
  \bibinfo{person}{Mohammad~Aflah Khan}, \bibinfo{person}{Shivanshu Purohit},
  \bibinfo{person}{USVSN~Sai Prashanth}, \bibinfo{person}{Edward Raff},
  \bibinfo{person}{Aviya Skowron}, \bibinfo{person}{Lintang Sutawika}, {and}
  \bibinfo{person}{Oskar Van Der~Wal}.} \bibinfo{year}{2023}\natexlab{}.
\newblock \showarticletitle{Pythia: A Suite for Analyzing Large Language Models
  Across Training and Scaling}. In \bibinfo{booktitle}{\emph{Proceedings of the
  40th International Conference on Machine Learning}}
  \emph{(\bibinfo{series}{Proceedings of Machine Learning Research},
  Vol.~\bibinfo{volume}{202})}. \bibinfo{publisher}{PMLR},
  \bibinfo{pages}{2397--2430}.
\newblock


\bibitem[Bingmann et~al\mbox{.}(2016)]%
        {bingmann16}
\bibfield{author}{\bibinfo{person}{Timo Bingmann}, \bibinfo{person}{Johannes
  Fischer}, {and} \bibinfo{person}{Vitaly Osipov}.}
  \bibinfo{year}{2016}\natexlab{}.
\newblock \showarticletitle{Inducing Suffix and LCP Arrays in External Memory}.
\newblock \bibinfo{journal}{\emph{ACM J. Exp. Algorithmics}}
  \bibinfo{volume}{21}, Article \bibinfo{articleno}{2.3} (\bibinfo{date}{Sept.}
  \bibinfo{year}{2016}), \bibinfo{numpages}{27}~pages.
\newblock
\showISSN{1084-6654}
\href{https://doi.org/10.1145/2975593}{doi:\nolinkurl{10.1145/2975593}}


\bibitem[Blelloch et~al\mbox{.}(1991)]%
        {blelloch91}
\bibfield{author}{\bibinfo{person}{Guy~E. Blelloch},
  \bibinfo{person}{Charles~E. Leiserson}, \bibinfo{person}{Bruce~M. Maggs},
  \bibinfo{person}{C.~Greg Plaxton}, \bibinfo{person}{Stephen~J. Smith}, {and}
  \bibinfo{person}{Marco Zagha}.} \bibinfo{year}{1991}\natexlab{}.
\newblock \showarticletitle{A comparison of sorting algorithms for the
  {Connection Machine CM-2}}. In \bibinfo{booktitle}{\emph{Proceedings of the
  Third Annual ACM Symposium on Parallel Algorithms and Architectures}} (Hilton
  Head, South Carolina, USA) \emph{(\bibinfo{series}{SPAA '91})}.
  \bibinfo{publisher}{Association for Computing Machinery},
  \bibinfo{address}{New York, NY, USA}, \bibinfo{pages}{3--16}.
\newblock
\showISBNx{0897914384}
\href{https://doi.org/10.1145/113379.113380}{doi:\nolinkurl{10.1145/113379.113380}}


\bibitem[B\"{u}ren et~al\mbox{.}(2019)]%
        {buren19}
\bibfield{author}{\bibinfo{person}{Florian B\"{u}ren}, \bibinfo{person}{Daniel
  J\"{u}nger}, \bibinfo{person}{Robin Kobus}, \bibinfo{person}{Christian
  Hundt}, {and} \bibinfo{person}{Bertil Schmidt}.}
  \bibinfo{year}{2019}\natexlab{}.
\newblock \showarticletitle{Suffix Array Construction on Multi-GPU Systems}. In
  \bibinfo{booktitle}{\emph{Proceedings of the 28th International Symposium on
  High-Performance Parallel and Distributed Computing}} (Phoenix, AZ, USA)
  \emph{(\bibinfo{series}{HPDC '19})}. \bibinfo{publisher}{Association for
  Computing Machinery}, \bibinfo{address}{New York, NY, USA},
  \bibinfo{pages}{183--194}.
\newblock
\showISBNx{9781450366700}
\href{https://doi.org/10.1145/3307681.3325961}{doi:\nolinkurl{10.1145/3307681.3325961}}


\bibitem[Cenzato et~al\mbox{.}(2026)]%
        {cenzato26}
\bibfield{author}{\bibinfo{person}{Davide Cenzato}, \bibinfo{person}{Lore
  Depuydt}, \bibinfo{person}{Travis Gagie}, \bibinfo{person}{Sung-Hwan Kim},
  \bibinfo{person}{Giovanni Manzini}, \bibinfo{person}{Francisco Olivares},
  {and} \bibinfo{person}{Nicola Prezza}.} \bibinfo{year}{2026}\natexlab{}.
\newblock \showarticletitle{Suffixient Arrays: A New Efficient Suffix Array
  Compression Technique}.
\newblock \bibinfo{journal}{\emph{Theory of Computing Systems}}
  \bibinfo{volume}{70}, \bibinfo{number}{3} (\bibinfo{date}{09}
  \bibinfo{year}{2026}), \bibinfo{pages}{50}.
\newblock
\href{https://doi.org/10.1007/s00224-026-10287-6}{doi:\nolinkurl{10.1007/s00224-026-10287-6}}


\bibitem[Dementiev et~al\mbox{.}(2008)]%
        {dementiev08}
\bibfield{author}{\bibinfo{person}{Roman Dementiev}, \bibinfo{person}{Juha
  K{\"a}rkk{\"a}inen}, \bibinfo{person}{Jens Mehnert}, {and}
  \bibinfo{person}{Peter Sanders}.} \bibinfo{year}{2008}\natexlab{}.
\newblock \showarticletitle{Better external memory suffix array construction}.
\newblock \bibinfo{journal}{\emph{ACM J. Exp. Algorithmics}}
  \bibinfo{volume}{12}, Article \bibinfo{articleno}{3.4}
  (\bibinfo{year}{2008}), \bibinfo{numpages}{24}~pages.
\newblock
\href{https://doi.org/10.1145/1227161.1402296}{doi:\nolinkurl{10.1145/1227161.1402296}}


\bibitem[Ferragina(2008)]%
        {ferragina08}
\bibfield{author}{\bibinfo{person}{Paolo Ferragina}.}
  \bibinfo{year}{2008}\natexlab{}.
\newblock \showarticletitle{String algorithms and data structures}.
\newblock \bibinfo{journal}{\emph{CoRR}}  \bibinfo{volume}{abs/0801.2378}
  (\bibinfo{year}{2008}).
\newblock
\showeprint[arXiv]{0801.2378}
\href{https://doi.org/10.48550/ARXIV.0801.2378}{doi:\nolinkurl{10.48550/ARXIV.0801.2378}}


\bibitem[Flick and Aluru(2015)]%
        {flick15}
\bibfield{author}{\bibinfo{person}{Patrick Flick} {and}
  \bibinfo{person}{Srinivas Aluru}.} \bibinfo{year}{2015}\natexlab{}.
\newblock \showarticletitle{Parallel distributed memory construction of suffix
  and longest common prefix arrays}. In \bibinfo{booktitle}{\emph{Proceedings
  of the International Conference for High Performance Computing, Networking,
  Storage and Analysis}} (Austin, Texas) \emph{(\bibinfo{series}{SC '15})}.
  \bibinfo{publisher}{Association for Computing Machinery},
  \bibinfo{address}{New York, NY, USA}, Article \bibinfo{articleno}{16},
  \bibinfo{numpages}{10}~pages.
\newblock
\showISBNx{9781450337236}
\href{https://doi.org/10.1145/2807591.2807609}{doi:\nolinkurl{10.1145/2807591.2807609}}


\bibitem[Frazer and McKellar(1970)]%
        {frazer70}
\bibfield{author}{\bibinfo{person}{W.~D. Frazer} {and} \bibinfo{person}{A.~C.
  McKellar}.} \bibinfo{year}{1970}\natexlab{}.
\newblock \showarticletitle{Samplesort: A Sampling Approach to Minimal Storage
  Tree Sorting}.
\newblock \bibinfo{journal}{\emph{J. ACM}} \bibinfo{volume}{17},
  \bibinfo{number}{3} (\bibinfo{date}{July} \bibinfo{year}{1970}),
  \bibinfo{pages}{496--507}.
\newblock
\showISSN{0004-5411}
\href{https://doi.org/10.1145/321592.321600}{doi:\nolinkurl{10.1145/321592.321600}}


\bibitem[Gao et~al\mbox{.}(2020)]%
        {gao20}
\bibfield{author}{\bibinfo{person}{Leo Gao}, \bibinfo{person}{Stella Biderman},
  \bibinfo{person}{Sid Black}, \bibinfo{person}{Laurence Golding},
  \bibinfo{person}{Travis Hoppe}, \bibinfo{person}{Charles Foster},
  \bibinfo{person}{Jason Phang}, \bibinfo{person}{Horace He},
  \bibinfo{person}{Anish Thite}, \bibinfo{person}{Noa Nabeshima},
  \bibinfo{person}{Shawn Presser}, {and} \bibinfo{person}{Connor Leahy}.}
  \bibinfo{year}{2020}\natexlab{}.
\newblock \showarticletitle{The Pile: An 800GB Dataset of Diverse Text for
  Language Modeling}.
\newblock \bibinfo{journal}{\emph{arXiv preprint arXiv:2101.00027}}
  (\bibinfo{year}{2020}).
\newblock


\bibitem[Ghoting and Makarychev(2009)]%
        {ghoting09}
\bibfield{author}{\bibinfo{person}{Amol Ghoting} {and}
  \bibinfo{person}{Konstantin Makarychev}.} \bibinfo{year}{2009}\natexlab{}.
\newblock \showarticletitle{Serial and parallel methods for {I/O} efficient
  suffix tree construction}. In \bibinfo{booktitle}{\emph{Proceedings of the
  2009 ACM SIGMOD International Conference on Management of Data}} (Providence,
  Rhode Island, USA) \emph{(\bibinfo{series}{SIGMOD '09})}.
  \bibinfo{publisher}{Association for Computing Machinery},
  \bibinfo{address}{New York, NY, USA}, \bibinfo{pages}{827--840}.
\newblock
\showISBNx{9781605585512}
\href{https://doi.org/10.1145/1559845.1559931}{doi:\nolinkurl{10.1145/1559845.1559931}}


\bibitem[Grebnov(2023)]%
        {libcubwt}
\bibfield{author}{\bibinfo{person}{Ilya Grebnov}.}
  \bibinfo{year}{2023}\natexlab{}.
\newblock \bibinfo{booktitle}{\emph{libcubwt}}.
\newblock
\urldef\tempurl%
\url{https://github.com/IlyaGrebnov/libcubwt}
\showURL{%
Retrieved April 20, 2026 from \tempurl}
\newblock
\shownote{Version 1.0.0}.


\bibitem[Grebnov(2025)]%
        {libsais}
\bibfield{author}{\bibinfo{person}{Ilya Grebnov}.}
  \bibinfo{year}{2025}\natexlab{}.
\newblock \bibinfo{booktitle}{\emph{libsais}}.
\newblock
\urldef\tempurl%
\url{https://github.com/IlyaGrebnov/libsais}
\showURL{%
Retrieved April 20, 2026 from \tempurl}
\newblock
\shownote{Version 2.10.4}.


\bibitem[He et~al\mbox{.}(2024)]%
        {he24}
\bibfield{author}{\bibinfo{person}{Zhenyu He}, \bibinfo{person}{Zexuan Zhong},
  \bibinfo{person}{Tianle Cai}, \bibinfo{person}{Jason Lee}, {and}
  \bibinfo{person}{Di He}.} \bibinfo{year}{2024}\natexlab{}.
\newblock \showarticletitle{{REST}: Retrieval-Based Speculative Decoding}. In
  \bibinfo{booktitle}{\emph{Proceedings of the 2024 Conference of the North
  American Chapter of the Association for Computational Linguistics: Human
  Language Technologies (Volume 1: Long Papers)}}.
  \bibinfo{publisher}{Association for Computational Linguistics},
  \bibinfo{address}{Mexico City, Mexico}, \bibinfo{pages}{1582--1595}.
\newblock
\href{https://doi.org/10.18653/v1/2024.naacl-long.88}{doi:\nolinkurl{10.18653/v1/2024.naacl-long.88}}


\bibitem[Ho et~al\mbox{.}(2024)]%
        {ho24}
\bibfield{author}{\bibinfo{person}{Sophia Ho}, \bibinfo{person}{Jinsol Park},
  {and} \bibinfo{person}{Patrick Wang}.} \bibinfo{year}{2024}\natexlab{}.
\newblock \bibinfo{title}{{CREST}: Effectively Compacting a Datastore For
  Retrieval-Based Speculative Decoding}.
\newblock
\showeprint[arxiv]{2408.04678}~[cs.CL]


\bibitem[Kandpal et~al\mbox{.}(2022)]%
        {kandpal22}
\bibfield{author}{\bibinfo{person}{Nikhil Kandpal}, \bibinfo{person}{Eric
  Wallace}, {and} \bibinfo{person}{Colin Raffel}.}
  \bibinfo{year}{2022}\natexlab{}.
\newblock \showarticletitle{Deduplicating Training Data Mitigates Privacy Risks
  in Language Models}. In \bibinfo{booktitle}{\emph{Proceedings of the 39th
  International Conference on Machine Learning}}
  \emph{(\bibinfo{series}{Proceedings of Machine Learning Research},
  Vol.~\bibinfo{volume}{162})}. \bibinfo{publisher}{PMLR},
  \bibinfo{address}{Baltimore, Maryland, USA}, \bibinfo{pages}{10697--10707}.
\newblock


\bibitem[K{\"a}rkk{\"a}inen et~al\mbox{.}(2015)]%
        {karkkainen15}
\bibfield{author}{\bibinfo{person}{Juha K{\"a}rkk{\"a}inen},
  \bibinfo{person}{Dominik Kempa}, {and} \bibinfo{person}{Simon~J. Puglisi}.}
  \bibinfo{year}{2015}\natexlab{}.
\newblock \showarticletitle{Parallel External Memory Suffix Sorting}. In
  \bibinfo{booktitle}{\emph{Combinatorial Pattern Matching}}.
  \bibinfo{publisher}{Springer International Publishing},
  \bibinfo{address}{Cham}, \bibinfo{pages}{329--342}.
\newblock
\href{https://doi.org/10.1007/978-3-319-19929-0_28}{doi:\nolinkurl{10.1007/978-3-319-19929-0_28}}


\bibitem[K{\"a}rkk{\"a}inen et~al\mbox{.}(2009)]%
        {karkkainen09}
\bibfield{author}{\bibinfo{person}{Juha K{\"a}rkk{\"a}inen},
  \bibinfo{person}{Giovanni Manzini}, {and} \bibinfo{person}{Simon~J.
  Puglisi}.} \bibinfo{year}{2009}\natexlab{}.
\newblock \showarticletitle{Permuted Longest-Common-Prefix Array}. In
  \bibinfo{booktitle}{\emph{Combinatorial Pattern Matching}},
  \bibfield{editor}{\bibinfo{person}{Gregory Kucherov} {and}
  \bibinfo{person}{Esko Ukkonen}} (Eds.). \bibinfo{publisher}{Springer Berlin
  Heidelberg}, \bibinfo{address}{Berlin, Heidelberg},
  \bibinfo{pages}{181--192}.
\newblock
\showISBNx{978-3-642-02441-2}
\href{https://doi.org/10.1007/978-3-642-02441-2_17}{doi:\nolinkurl{10.1007/978-3-642-02441-2_17}}


\bibitem[Khan et~al\mbox{.}(2023)]%
        {khan23}
\bibfield{author}{\bibinfo{person}{Jamshed Khan}, \bibinfo{person}{Tobias
  Rubel}, \bibinfo{person}{Laxman Dhulipala}, \bibinfo{person}{Erin Molloy},
  {and} \bibinfo{person}{Rob Patro}.} \bibinfo{year}{2023}\natexlab{}.
\newblock \showarticletitle{{Fast, Parallel, and Cache-Friendly Suffix Array
  Construction}}. In \bibinfo{booktitle}{\emph{23rd International Workshop on
  Algorithms in Bioinformatics (WABI 2023)}} \emph{(\bibinfo{series}{Leibniz
  International Proceedings in Informatics (LIPIcs)},
  Vol.~\bibinfo{volume}{273})}, \bibfield{editor}{\bibinfo{person}{Djamal
  Belazzougui} {and} \bibinfo{person}{A\"{i}da Ouangraoua}} (Eds.).
  \bibinfo{publisher}{Schloss Dagstuhl -- Leibniz-Zentrum f{\"u}r Informatik},
  \bibinfo{address}{Dagstuhl, Germany}, \bibinfo{pages}{16:1--16:21}.
\newblock
\showISBNx{978-3-95977-294-5}
\showISSN{1868-8969}
\href{https://doi.org/10.4230/LIPIcs.WABI.2023.16}{doi:\nolinkurl{10.4230/LIPIcs.WABI.2023.16}}


\bibitem[Ko and Aluru(2003)]%
        {ko03}
\bibfield{author}{\bibinfo{person}{Pang Ko} {and} \bibinfo{person}{Srinivas
  Aluru}.} \bibinfo{year}{2003}\natexlab{}.
\newblock \showarticletitle{Space Efficient Linear Time Construction of Suffix
  Arrays}. In \bibinfo{booktitle}{\emph{Combinatorial Pattern Matching}},
  \bibfield{editor}{\bibinfo{person}{Ricardo Baeza-Yates},
  \bibinfo{person}{Edgar Ch{\'a}vez}, {and} \bibinfo{person}{Maxime
  Crochemore}} (Eds.). \bibinfo{publisher}{Springer Berlin Heidelberg},
  \bibinfo{address}{Berlin, Heidelberg}, \bibinfo{pages}{200--210}.
\newblock
\showISBNx{978-3-540-44888-4}
\href{https://doi.org/10.1007/3-540-44888-8_15}{doi:\nolinkurl{10.1007/3-540-44888-8_15}}


\bibitem[Kocberber et~al\mbox{.}(2015)]%
        {kocberber15}
\bibfield{author}{\bibinfo{person}{Onur Kocberber}, \bibinfo{person}{Babak
  Falsafi}, {and} \bibinfo{person}{Boris Grot}.}
  \bibinfo{year}{2015}\natexlab{}.
\newblock \showarticletitle{Asynchronous memory access chaining}.
\newblock \bibinfo{journal}{\emph{Proc. VLDB Endow.}} \bibinfo{volume}{9},
  \bibinfo{number}{4} (\bibinfo{date}{Dec.} \bibinfo{year}{2015}),
  \bibinfo{pages}{252--263}.
\newblock
\showISSN{2150-8097}
\href{https://doi.org/10.14778/2856318.2856321}{doi:\nolinkurl{10.14778/2856318.2856321}}


\bibitem[Kulla and Sanders(2007)]%
        {kulla07}
\bibfield{author}{\bibinfo{person}{Fabian Kulla} {and} \bibinfo{person}{Peter
  Sanders}.} \bibinfo{year}{2007}\natexlab{}.
\newblock \showarticletitle{Scalable parallel suffix array construction}.
\newblock \bibinfo{journal}{\emph{Parallel Comput.}} \bibinfo{volume}{33},
  \bibinfo{number}{9} (\bibinfo{year}{2007}), \bibinfo{pages}{605--612}.
\newblock
\showISSN{0167-8191}
\href{https://doi.org/10.1016/j.parco.2007.06.004}{doi:\nolinkurl{10.1016/j.parco.2007.06.004}}
\newblock
\shownote{Selected Papers from EuroPVM/MPI 2006}.


\bibitem[Kärkkäinen(2007)]%
        {karkkainen07}
\bibfield{author}{\bibinfo{person}{Juha Kärkkäinen}.}
  \bibinfo{year}{2007}\natexlab{}.
\newblock \showarticletitle{Fast BWT in small space by blockwise suffix
  sorting}.
\newblock \bibinfo{journal}{\emph{Theoretical Computer Science}}
  \bibinfo{volume}{387}, \bibinfo{number}{3} (\bibinfo{year}{2007}),
  \bibinfo{pages}{249--257}.
\newblock
\showISSN{0304-3975}
\href{https://doi.org/10.1016/j.tcs.2007.07.018}{doi:\nolinkurl{10.1016/j.tcs.2007.07.018}}


\bibitem[Larsson and Sadakane(2007)]%
        {larsson07}
\bibfield{author}{\bibinfo{person}{N.~Jesper Larsson} {and}
  \bibinfo{person}{Kunihiko Sadakane}.} \bibinfo{year}{2007}\natexlab{}.
\newblock \showarticletitle{Faster suffix sorting}.
\newblock \bibinfo{journal}{\emph{Theoretical Computer Science}}
  \bibinfo{volume}{387}, \bibinfo{number}{3} (\bibinfo{year}{2007}),
  \bibinfo{pages}{258--272}.
\newblock
\showISSN{0304-3975}
\href{https://doi.org/10.1016/j.tcs.2007.07.017}{doi:\nolinkurl{10.1016/j.tcs.2007.07.017}}


\bibitem[Lee et~al\mbox{.}(2022)]%
        {lee22}
\bibfield{author}{\bibinfo{person}{Katherine Lee}, \bibinfo{person}{Daphne
  Ippolito}, \bibinfo{person}{Andrew Nystrom}, \bibinfo{person}{Chiyuan Zhang},
  \bibinfo{person}{Douglas Eck}, \bibinfo{person}{Chris Callison-Burch}, {and}
  \bibinfo{person}{Nicholas Carlini}.} \bibinfo{year}{2022}\natexlab{}.
\newblock \showarticletitle{Deduplicating Training Data Makes Language Models
  Better}. In \bibinfo{booktitle}{\emph{Proceedings of the 60th Annual Meeting
  of the Association for Computational Linguistics (Volume 1: Long Papers)}}.
  \bibinfo{publisher}{Association for Computational Linguistics},
  \bibinfo{pages}{8424--8445}.
\newblock
\href{https://doi.org/10.18653/v1/2022.acl-long.577}{doi:\nolinkurl{10.18653/v1/2022.acl-long.577}}


\bibitem[Liu(2025)]%
        {infinigram}
\bibfield{author}{\bibinfo{person}{Jiacheng Liu}.}
  \bibinfo{year}{2025}\natexlab{}.
\newblock \bibinfo{booktitle}{\emph{infini-gram}}.
\newblock
\urldef\tempurl%
\url{https://github.com/liujch1998/infini-gram}
\showURL{%
Retrieved October 14, 2026 from \tempurl}
\newblock
\shownote{Version 2.6.0}.


\bibitem[Liu et~al\mbox{.}(2025)]%
        {liu25}
\bibfield{author}{\bibinfo{person}{Jiacheng Liu}, \bibinfo{person}{Taylor
  Blanton}, \bibinfo{person}{Yanai Elazar}, \bibinfo{person}{Sewon Min},
  \bibinfo{person}{Yen-Sung Chen}, \bibinfo{person}{Arnavi Chheda-Kothary},
  \bibinfo{person}{Huy Tran}, \bibinfo{person}{Byron Bischoff},
  \bibinfo{person}{Eric Marsh}, \bibinfo{person}{Michael Schmitz},
  \bibinfo{person}{Cassidy Trier}, \bibinfo{person}{Aaron Sarnat},
  \bibinfo{person}{Jenna James}, \bibinfo{person}{Jon Borchardt},
  \bibinfo{person}{Bailey Kuehl}, \bibinfo{person}{Evie Yu-Yen Cheng},
  \bibinfo{person}{Karen Farley}, \bibinfo{person}{Taira Anderson},
  \bibinfo{person}{David Albright}, \bibinfo{person}{Carissa Schoenick},
  \bibinfo{person}{Luca Soldaini}, \bibinfo{person}{Dirk Groeneveld},
  \bibinfo{person}{Rock~Yuren Pang}, \bibinfo{person}{Pang~Wei Koh},
  \bibinfo{person}{Noah~A. Smith}, \bibinfo{person}{Sophie Lebrecht},
  \bibinfo{person}{Yejin Choi}, \bibinfo{person}{Hannaneh Hajishirzi},
  \bibinfo{person}{Ali Farhadi}, {and} \bibinfo{person}{Jesse Dodge}.}
  \bibinfo{year}{2025}\natexlab{}.
\newblock \showarticletitle{{OLM}o{T}race: Tracing Language Model Outputs Back
  to Trillions of Training Tokens}. In \bibinfo{booktitle}{\emph{Proceedings of
  the 63rd Annual Meeting of the Association for Computational Linguistics
  (Volume 3: System Demonstrations)}},
  \bibfield{editor}{\bibinfo{person}{Pushkar Mishra}, \bibinfo{person}{Smaranda
  Muresan}, {and} \bibinfo{person}{Tao Yu}} (Eds.).
  \bibinfo{publisher}{Association for Computational Linguistics},
  \bibinfo{address}{Vienna, Austria}, \bibinfo{pages}{178--188}.
\newblock
\showISBNx{979-8-89176-253-4}
\href{https://doi.org/10.18653/v1/2025.acl-demo.18}{doi:\nolinkurl{10.18653/v1/2025.acl-demo.18}}


\bibitem[Liu et~al\mbox{.}(2024)]%
        {liu24}
\bibfield{author}{\bibinfo{person}{Jiacheng Liu}, \bibinfo{person}{Sewon Min},
  \bibinfo{person}{Luke Zettlemoyer}, \bibinfo{person}{Yejin Choi}, {and}
  \bibinfo{person}{Hannaneh Hajishirzi}.} \bibinfo{year}{2024}\natexlab{}.
\newblock \showarticletitle{Infini-gram: Scaling Unbounded n-gram Language
  Models to a Trillion Tokens}. In \bibinfo{booktitle}{\emph{First Conference
  on Language Modeling}}.
\newblock
\urldef\tempurl%
\url{https://openreview.net/forum?id=u2vAyMeLMm}
\showURL{%
\tempurl}


\bibitem[Lutz et~al\mbox{.}(2020)]%
        {lutz20}
\bibfield{author}{\bibinfo{person}{Clemens Lutz}, \bibinfo{person}{Sebastian
  Bre{\ss}}, \bibinfo{person}{Steffen Zeuch}, \bibinfo{person}{Tilmann Rabl},
  {and} \bibinfo{person}{Volker Markl}.} \bibinfo{year}{2020}\natexlab{}.
\newblock \showarticletitle{Pump Up the Volume: Processing Large Data on {GPUs}
  with Fast Interconnects}. In \bibinfo{booktitle}{\emph{Proceedings of the
  2020 ACM SIGMOD International Conference on Management of Data}}
  \emph{(\bibinfo{series}{SIGMOD '20})}. \bibinfo{publisher}{Association for
  Computing Machinery}, \bibinfo{address}{New York, NY, USA},
  \bibinfo{pages}{1633--1649}.
\newblock
\href{https://doi.org/10.1145/3318464.3389705}{doi:\nolinkurl{10.1145/3318464.3389705}}


\bibitem[Lutz et~al\mbox{.}(2022)]%
        {lutz22}
\bibfield{author}{\bibinfo{person}{Clemens Lutz}, \bibinfo{person}{Sebastian
  Bre{\ss}}, \bibinfo{person}{Steffen Zeuch}, \bibinfo{person}{Tilmann Rabl},
  {and} \bibinfo{person}{Volker Markl}.} \bibinfo{year}{2022}\natexlab{}.
\newblock \showarticletitle{Triton Join: Efficiently Scaling to a Large Join
  State on {GPUs} with Fast Interconnects}. In
  \bibinfo{booktitle}{\emph{Proceedings of the 2022 International Conference on
  Management of Data}} \emph{(\bibinfo{series}{SIGMOD '22})}.
  \bibinfo{publisher}{Association for Computing Machinery},
  \bibinfo{address}{New York, NY, USA}, \bibinfo{pages}{1017--1032}.
\newblock
\href{https://doi.org/10.1145/3514221.3517911}{doi:\nolinkurl{10.1145/3514221.3517911}}


\bibitem[Maltenberger et~al\mbox{.}(2025)]%
        {maltenberger25}
\bibfield{author}{\bibinfo{person}{Tobias Maltenberger}, \bibinfo{person}{Ilin
  Tolovski}, {and} \bibinfo{person}{Tilmann Rabl}.}
  \bibinfo{year}{2025}\natexlab{}.
\newblock \showarticletitle{Efficiently Joining Large Relations on Multi-GPU
  Systems}.
\newblock \bibinfo{journal}{\emph{Proc. VLDB Endow.}} \bibinfo{volume}{18},
  \bibinfo{number}{11} (\bibinfo{date}{July} \bibinfo{year}{2025}),
  \bibinfo{pages}{4653--4667}.
\newblock
\showISSN{2150-8097}
\href{https://doi.org/10.14778/3749646.3749720}{doi:\nolinkurl{10.14778/3749646.3749720}}


\bibitem[Manber and Myers(1990)]%
        {manber90}
\bibfield{author}{\bibinfo{person}{Udi Manber} {and} \bibinfo{person}{Gene
  Myers}.} \bibinfo{year}{1990}\natexlab{}.
\newblock \showarticletitle{Suffix arrays: a new method for on-line string
  searches}. In \bibinfo{booktitle}{\emph{Proceedings of the First Annual
  ACM-SIAM Symposium on Discrete Algorithms}} (San Francisco, California, USA)
  \emph{(\bibinfo{series}{SODA '90})}. \bibinfo{publisher}{Society for
  Industrial and Applied Mathematics}, \bibinfo{address}{USA},
  \bibinfo{pages}{319--327}.
\newblock
\showISBNx{0898712513}


\bibitem[Mansour et~al\mbox{.}(2011)]%
        {mansour11}
\bibfield{author}{\bibinfo{person}{Essam Mansour}, \bibinfo{person}{Amin
  Allam}, \bibinfo{person}{Spiros Skiadopoulos}, {and} \bibinfo{person}{Panos
  Kalnis}.} \bibinfo{year}{2011}\natexlab{}.
\newblock \showarticletitle{{ERA}: efficient serial and parallel suffix tree
  construction for very long strings}.
\newblock \bibinfo{journal}{\emph{Proc. VLDB Endow.}} \bibinfo{volume}{5},
  \bibinfo{number}{1} (\bibinfo{year}{2011}), \bibinfo{pages}{49--60}.
\newblock
\href{https://doi.org/10.14778/2047485.2047490}{doi:\nolinkurl{10.14778/2047485.2047490}}


\bibitem[Merrill and Grimshaw(2011)]%
        {merrill11}
\bibfield{author}{\bibinfo{person}{Duane Merrill} {and} \bibinfo{person}{Andrew
  Grimshaw}.} \bibinfo{year}{2011}\natexlab{}.
\newblock \showarticletitle{High Performance and Scalable Radix Sorting: A Case
  Study of Implementing Dynamic Parallelism for {GPU} Computing}.
\newblock \bibinfo{journal}{\emph{Parallel Processing Letters}}
  \bibinfo{volume}{21}, \bibinfo{number}{2} (\bibinfo{year}{2011}),
  \bibinfo{pages}{245--272}.
\newblock
\href{https://doi.org/10.1142/S0129626411000187}{doi:\nolinkurl{10.1142/S0129626411000187}}


\bibitem[Min et~al\mbox{.}(2020)]%
        {min20}
\bibfield{author}{\bibinfo{person}{Seung~Won Min},
  \bibinfo{person}{Vikram~Sharma Mailthody}, \bibinfo{person}{Zaid Qureshi},
  \bibinfo{person}{Jinjun Xiong}, \bibinfo{person}{Eiman Ebrahimi}, {and}
  \bibinfo{person}{Wen-mei Hwu}.} \bibinfo{year}{2020}\natexlab{}.
\newblock \showarticletitle{{EMOGI}: Efficient Memory-access for Out-of-memory
  Graph-traversal in {GPUs}}.
\newblock \bibinfo{journal}{\emph{Proc. VLDB Endow.}} \bibinfo{volume}{14},
  \bibinfo{number}{2} (\bibinfo{year}{2020}), \bibinfo{pages}{114--127}.
\newblock
\href{https://doi.org/10.14778/3425879.3425883}{doi:\nolinkurl{10.14778/3425879.3425883}}


\bibitem[Nong et~al\mbox{.}(2011)]%
        {nong11}
\bibfield{author}{\bibinfo{person}{Ge Nong}, \bibinfo{person}{Sen Zhang}, {and}
  \bibinfo{person}{Wai~Hong Chan}.} \bibinfo{year}{2011}\natexlab{}.
\newblock \showarticletitle{Two Efficient Algorithms for Linear Time Suffix
  Array Construction}.
\newblock \bibinfo{journal}{\emph{IEEE Trans. Comput.}} \bibinfo{volume}{60},
  \bibinfo{number}{10} (\bibinfo{date}{Oct.} \bibinfo{year}{2011}),
  \bibinfo{pages}{1471--1484}.
\newblock
\showISSN{0018-9340}
\href{https://doi.org/10.1109/TC.2010.188}{doi:\nolinkurl{10.1109/TC.2010.188}}


\bibitem[Osipov(2012)]%
        {osipov12}
\bibfield{author}{\bibinfo{person}{Vitaly Osipov}.}
  \bibinfo{year}{2012}\natexlab{}.
\newblock \showarticletitle{Parallel Suffix Array Construction for Shared
  Memory Architectures}. In \bibinfo{booktitle}{\emph{String Processing and
  Information Retrieval}}, \bibfield{editor}{\bibinfo{person}{Liliana
  Calder{\'o}n-Benavides}, \bibinfo{person}{Cristina Gonz{\'a}lez-Caro},
  \bibinfo{person}{Edgar Ch{\'a}vez}, {and} \bibinfo{person}{Nivio Ziviani}}
  (Eds.). \bibinfo{publisher}{Springer Berlin Heidelberg},
  \bibinfo{address}{Berlin, Heidelberg}, \bibinfo{pages}{379--384}.
\newblock
\showISBNx{978-3-642-34109-0}
\href{https://doi.org/10.1007/978-3-642-34109-0_40}{doi:\nolinkurl{10.1007/978-3-642-34109-0_40}}


\bibitem[Peng et~al\mbox{.}(2026)]%
        {peng26}
\bibfield{author}{\bibinfo{person}{Yuchen Peng}, \bibinfo{person}{Dingyu Yang},
  \bibinfo{person}{Zhongle Xie}, \bibinfo{person}{Ji Sun},
  \bibinfo{person}{Lidan Shou}, \bibinfo{person}{Ke Chen}, {and}
  \bibinfo{person}{Gang Chen}.} \bibinfo{year}{2026}\natexlab{}.
\newblock \showarticletitle{{SVFusion}: A {CPU-GPU} Co-Processing Architecture
  for Large-Scale Real-Time Vector Search}.
\newblock \bibinfo{journal}{\emph{Proc. VLDB Endow.}} \bibinfo{volume}{19},
  \bibinfo{number}{5} (\bibinfo{year}{2026}), \bibinfo{pages}{1074--1087}.
\newblock
\href{https://doi.org/10.14778/3796195.3796216}{doi:\nolinkurl{10.14778/3796195.3796216}}


\bibitem[Raffel et~al\mbox{.}(2020)]%
        {raffel20}
\bibfield{author}{\bibinfo{person}{Colin Raffel}, \bibinfo{person}{Noam
  Shazeer}, \bibinfo{person}{Adam Roberts}, \bibinfo{person}{Katherine Lee},
  \bibinfo{person}{Sharan Narang}, \bibinfo{person}{Michael Matena},
  \bibinfo{person}{Yanqi Zhou}, \bibinfo{person}{Wei Li}, {and}
  \bibinfo{person}{Peter~J. Liu}.} \bibinfo{year}{2020}\natexlab{}.
\newblock \showarticletitle{Exploring the Limits of Transfer Learning with a
  Unified Text-to-Text Transformer}.
\newblock \bibinfo{journal}{\emph{Journal of Machine Learning Research}}
  \bibinfo{volume}{21}, \bibinfo{number}{140} (\bibinfo{year}{2020}),
  \bibinfo{pages}{1--67}.
\newblock
\urldef\tempurl%
\url{http://jmlr.org/papers/v21/20-074.html}
\showURL{%
\tempurl}


\bibitem[Roy et~al\mbox{.}(2013)]%
        {roy13}
\bibfield{author}{\bibinfo{person}{Amitabha Roy}, \bibinfo{person}{Ivo
  Mihailovic}, {and} \bibinfo{person}{Willy Zwaenepoel}.}
  \bibinfo{year}{2013}\natexlab{}.
\newblock \showarticletitle{{X-Stream}: edge-centric graph processing using
  streaming partitions}. In \bibinfo{booktitle}{\emph{Proceedings of the
  Twenty-Fourth ACM Symposium on Operating Systems Principles}} (Farmington,
  Pennsylvania) \emph{(\bibinfo{series}{SOSP '13})}.
  \bibinfo{publisher}{Association for Computing Machinery},
  \bibinfo{address}{New York, NY, USA}, \bibinfo{pages}{472--488}.
\newblock
\showISBNx{9781450323888}
\href{https://doi.org/10.1145/2517349.2522740}{doi:\nolinkurl{10.1145/2517349.2522740}}


\bibitem[Soldaini et~al\mbox{.}(2024)]%
        {soldaini24}
\bibfield{author}{\bibinfo{person}{Luca Soldaini}, \bibinfo{person}{Rodney
  Kinney}, \bibinfo{person}{Akshita Bhagia}, \bibinfo{person}{Dustin Schwenk},
  \bibinfo{person}{David Atkinson}, \bibinfo{person}{Russell Authur},
  \bibinfo{person}{Ben Bogin}, \bibinfo{person}{Khyathi Chandu},
  \bibinfo{person}{Jennifer Dumas}, \bibinfo{person}{Yanai Elazar},
  \bibinfo{person}{Valentin Hofmann}, \bibinfo{person}{Ananya Jha},
  \bibinfo{person}{Sachin Kumar}, \bibinfo{person}{Li Lucy},
  \bibinfo{person}{Xinxi Lyu}, \bibinfo{person}{Nathan Lambert},
  \bibinfo{person}{Ian Magnusson}, \bibinfo{person}{Jacob Morrison},
  \bibinfo{person}{Niklas Muennighoff}, \bibinfo{person}{Aakanksha Naik},
  \bibinfo{person}{Crystal Nam}, \bibinfo{person}{Matthew Peters},
  \bibinfo{person}{Abhilasha Ravichander}, \bibinfo{person}{Kyle Richardson},
  \bibinfo{person}{Zejiang Shen}, \bibinfo{person}{Emma Strubell},
  \bibinfo{person}{Nishant Subramani}, \bibinfo{person}{Oyvind Tafjord},
  \bibinfo{person}{Evan Walsh}, \bibinfo{person}{Luke Zettlemoyer},
  \bibinfo{person}{Noah Smith}, \bibinfo{person}{Hannaneh Hajishirzi},
  \bibinfo{person}{Iz Beltagy}, \bibinfo{person}{Dirk Groeneveld},
  \bibinfo{person}{Jesse Dodge}, {and} \bibinfo{person}{Kyle Lo}.}
  \bibinfo{year}{2024}\natexlab{}.
\newblock \showarticletitle{Dolma: an Open Corpus of Three Trillion Tokens for
  Language Model Pretraining Research}. In
  \bibinfo{booktitle}{\emph{Proceedings of the 62nd Annual Meeting of the
  Association for Computational Linguistics (Volume 1: Long Papers)}},
  \bibfield{editor}{\bibinfo{person}{Lun-Wei Ku}, \bibinfo{person}{Andre
  Martins}, {and} \bibinfo{person}{Vivek Srikumar}} (Eds.).
  \bibinfo{publisher}{Association for Computational Linguistics},
  \bibinfo{address}{Bangkok, Thailand}, \bibinfo{pages}{15725--15788}.
\newblock
\href{https://doi.org/10.18653/v1/2024.acl-long.840}{doi:\nolinkurl{10.18653/v1/2024.acl-long.840}}


\bibitem[Timoshevskaya and Feng(2014)]%
        {timoshevskaya14}
\bibfield{author}{\bibinfo{person}{Nataliya Timoshevskaya} {and}
  \bibinfo{person}{Wu-chun Feng}.} \bibinfo{year}{2014}\natexlab{}.
\newblock \showarticletitle{{SAIS-OPT}: On the characterization and
  optimization of the {SA-IS} algorithm for suffix array construction}. In
  \bibinfo{booktitle}{\emph{2014 IEEE 4th International Conference on
  Computational Advances in Bio and Medical Sciences (ICCABS)}}.
  \bibinfo{pages}{1--6}.
\newblock
\href{https://doi.org/10.1109/ICCABS.2014.6863917}{doi:\nolinkurl{10.1109/ICCABS.2014.6863917}}


\bibitem[Touvron et~al\mbox{.}(2023)]%
        {touvron23}
\bibfield{author}{\bibinfo{person}{Hugo Touvron}, \bibinfo{person}{Louis
  Martin}, \bibinfo{person}{Kevin Stone}, \bibinfo{person}{Peter Albert},
  \bibinfo{person}{Amjad Almahairi}, \bibinfo{person}{Yasmine Babaei},
  \bibinfo{person}{Nikolay Bashlykov}, \bibinfo{person}{Soumya Batra},
  \bibinfo{person}{Prajjwal Bhargava}, \bibinfo{person}{Shruti Bhosale},
  \bibinfo{person}{Dan Bikel}, \bibinfo{person}{Lukas Blecher},
  \bibinfo{person}{Cristian Canton~Ferrer}, \bibinfo{person}{Moya Chen},
  \bibinfo{person}{Guillem Cucurull}, \bibinfo{person}{David Esiobu},
  \bibinfo{person}{Jude Fernandes}, \bibinfo{person}{Jeremy Fu},
  \bibinfo{person}{Wenyin Fu}, \bibinfo{person}{Brian Fuller},
  \bibinfo{person}{Cynthia Gao}, \bibinfo{person}{Vedanuj Goswami},
  \bibinfo{person}{Naman Goyal}, \bibinfo{person}{Anthony Hartshorn},
  \bibinfo{person}{Saghar Hosseini}, \bibinfo{person}{Rui Hou},
  \bibinfo{person}{Hakan Inan}, \bibinfo{person}{Marcin Kardas},
  \bibinfo{person}{Viktor Kerkez}, \bibinfo{person}{Madian Khabsa},
  \bibinfo{person}{Isabel Kloumann}, \bibinfo{person}{Artem Korenev},
  \bibinfo{person}{Punit~Singh Koura}, \bibinfo{person}{Marie-Anne Lachaux},
  \bibinfo{person}{Thibaut Lavril}, \bibinfo{person}{Jenya Lee},
  \bibinfo{person}{Diana Liskovich}, \bibinfo{person}{Yinghai Lu},
  \bibinfo{person}{Yuning Mao}, \bibinfo{person}{Xavier Martinet},
  \bibinfo{person}{Todor Mihaylov}, \bibinfo{person}{Pushkar Mishra},
  \bibinfo{person}{Igor Molybog}, \bibinfo{person}{Yixin Nie},
  \bibinfo{person}{Andrew Poulton}, \bibinfo{person}{Jeremy Reizenstein},
  \bibinfo{person}{Rashi Rungta}, \bibinfo{person}{Kalyan Saladi},
  \bibinfo{person}{Alan Schelten}, \bibinfo{person}{Ruan Silva},
  \bibinfo{person}{Eric~Michael Smith}, \bibinfo{person}{Ranjan Subramanian},
  \bibinfo{person}{Xiaoqing~Ellen Tan}, \bibinfo{person}{Binh Tang},
  \bibinfo{person}{Ross Taylor}, \bibinfo{person}{Adina Williams},
  \bibinfo{person}{Jian~Xiang Kuan}, \bibinfo{person}{Puxin Xu},
  \bibinfo{person}{Zheng Yan}, \bibinfo{person}{Iliyan Zarov},
  \bibinfo{person}{Yuchen Zhang}, \bibinfo{person}{Angela Fan},
  \bibinfo{person}{Melanie Kambadur}, \bibinfo{person}{Sharan Narang},
  \bibinfo{person}{Aurelien Rodriguez}, \bibinfo{person}{Robert Stojnic},
  \bibinfo{person}{Sergey Edunov}, {and} \bibinfo{person}{Thomas Scialom}.}
  \bibinfo{year}{2023}\natexlab{}.
\newblock \showarticletitle{Llama 2: Open Foundation and Fine-Tuned Chat
  Models}.
\newblock \bibinfo{journal}{\emph{arXiv preprint arXiv:2307.09288}}
  (\bibinfo{year}{2023}).
\newblock


\bibitem[Wang et~al\mbox{.}(2026)]%
        {wang26}
\bibfield{author}{\bibinfo{person}{Hai Wang}, \bibinfo{person}{Chenhao Wang},
  \bibinfo{person}{Qifeng Cai}, \bibinfo{person}{Yixiu Liu},
  \bibinfo{person}{Miao Peng}, \bibinfo{person}{Nuo Chen},
  \bibinfo{person}{Yuanlin Tu}, \bibinfo{person}{Chengcheng Xu}, {and}
  \bibinfo{person}{Feng Zhang}.} \bibinfo{year}{2026}\natexlab{}.
\newblock \bibinfo{title}{Scalable Frequency- and Length-Aware Subdocument
  Deduplication for Large Language Model Pretraining}.
\newblock
\showeprint[arxiv]{2608.03089}~[cs.CL]


\bibitem[Wang et~al\mbox{.}(2016)]%
        {wang16}
\bibfield{author}{\bibinfo{person}{Leyuan Wang}, \bibinfo{person}{Sean Baxter},
  {and} \bibinfo{person}{John~D. Owens}.} \bibinfo{year}{2016}\natexlab{}.
\newblock \showarticletitle{Fast parallel skew and prefix-doubling suffix array
  construction on the GPU}.
\newblock \bibinfo{journal}{\emph{Concurr. Comput. : Pract. Exper.}}
  \bibinfo{volume}{28}, \bibinfo{number}{12} (\bibinfo{date}{aug}
  \bibinfo{year}{2016}), \bibinfo{pages}{3466--3484}.
\newblock
\showISSN{1532-0626}
\href{https://doi.org/10.1002/cpe.3867}{doi:\nolinkurl{10.1002/cpe.3867}}


\bibitem[Wang et~al\mbox{.}(2023)]%
        {wang23}
\bibfield{author}{\bibinfo{person}{Qiange Wang}, \bibinfo{person}{Xin Ai},
  \bibinfo{person}{Yanfeng Zhang}, \bibinfo{person}{Jing Chen}, {and}
  \bibinfo{person}{Ge Yu}.} \bibinfo{year}{2023}\natexlab{}.
\newblock \showarticletitle{{HyTGraph}: {GPU}-Accelerated Graph Processing with
  Hybrid Transfer Management}. In \bibinfo{booktitle}{\emph{2023 IEEE 39th
  International Conference on Data Engineering (ICDE)}}.
  \bibinfo{publisher}{IEEE}, \bibinfo{pages}{558--571}.
\newblock
\href{https://doi.org/10.1109/ICDE55515.2023.00049}{doi:\nolinkurl{10.1109/ICDE55515.2023.00049}}


\bibitem[Weber et~al\mbox{.}(2024)]%
        {weber24}
\bibfield{author}{\bibinfo{person}{Maurice Weber}, \bibinfo{person}{Daniel~Y.
  Fu}, \bibinfo{person}{Quentin Anthony}, \bibinfo{person}{Yonatan Oren},
  \bibinfo{person}{Shane Adams}, \bibinfo{person}{Anton Alexandrov},
  \bibinfo{person}{Xiaozhong Lyu}, \bibinfo{person}{Huu Nguyen},
  \bibinfo{person}{Xiaozhe Yao}, \bibinfo{person}{Virginia Adams},
  \bibinfo{person}{Ben Athiwaratkun}, \bibinfo{person}{Rahul Chalamala},
  \bibinfo{person}{Kezhen Chen}, \bibinfo{person}{Max Ryabinin},
  \bibinfo{person}{Tri Dao}, \bibinfo{person}{Percy Liang},
  \bibinfo{person}{Christopher R{\'e}}, \bibinfo{person}{Irina Rish}, {and}
  \bibinfo{person}{Ce Zhang}.} \bibinfo{year}{2024}\natexlab{}.
\newblock \showarticletitle{RedPajama: an Open Dataset for Training Large
  Language Models}. In \bibinfo{booktitle}{\emph{Advances in Neural Information
  Processing Systems}}, Vol.~\bibinfo{volume}{37}. \bibinfo{publisher}{Curran
  Associates, Inc.}, \bibinfo{pages}{116462--116492}.
\newblock
\href{https://doi.org/10.52202/079017-3697}{doi:\nolinkurl{10.52202/079017-3697}}


\bibitem[Xie et~al\mbox{.}(2020)]%
        {xie20}
\bibfield{author}{\bibinfo{person}{Jing~Yi Xie}, \bibinfo{person}{Ge Nong},
  \bibinfo{person}{Bin Lao}, {and} \bibinfo{person}{Wentao Xu}.}
  \bibinfo{year}{2020}\natexlab{}.
\newblock \showarticletitle{Scalable Suffix Sorting on a Multicore Machine}.
\newblock \bibinfo{journal}{\emph{IEEE Trans. Comput.}} \bibinfo{volume}{69},
  \bibinfo{number}{9} (\bibinfo{year}{2020}), \bibinfo{pages}{1364--1375}.
\newblock
\href{https://doi.org/10.1109/TC.2020.2972546}{doi:\nolinkurl{10.1109/TC.2020.2972546}}


\bibitem[Xu et~al\mbox{.}(2025)]%
        {xu25}
\bibfield{author}{\bibinfo{person}{Hao Xu}, \bibinfo{person}{Jiacheng Liu},
  \bibinfo{person}{Yejin Choi}, \bibinfo{person}{Noah~A. Smith}, {and}
  \bibinfo{person}{Hannaneh Hajishirzi}.} \bibinfo{year}{2025}\natexlab{}.
\newblock \showarticletitle{Infini-gram mini: Exact n-gram Search at the
  Internet Scale with {FM}-Index}. In \bibinfo{booktitle}{\emph{Proceedings of
  the 2025 Conference on Empirical Methods in Natural Language Processing}}.
  \bibinfo{publisher}{Association for Computational Linguistics},
  \bibinfo{address}{Suzhou, China}, \bibinfo{pages}{24944--24969}.
\newblock
\href{https://doi.org/10.18653/v1/2025.emnlp-main.1268}{doi:\nolinkurl{10.18653/v1/2025.emnlp-main.1268}}


\bibitem[Yoneda et~al\mbox{.}(2026)]%
        {yoneda26}
\bibfield{author}{\bibinfo{person}{Masataka Yoneda}, \bibinfo{person}{Yusuke
  Matsushita}, \bibinfo{person}{Go Kamoda}, \bibinfo{person}{Kohei Suenaga},
  \bibinfo{person}{Takuya Akiba}, \bibinfo{person}{Masaki Waga}, {and}
  \bibinfo{person}{Sho Yokoi}.} \bibinfo{year}{2026}\natexlab{}.
\newblock \showarticletitle{SoftMatcha 2: A Fast and Soft Pattern Matcher for
  Trillion-Scale Corpora}. In \bibinfo{booktitle}{\emph{Forty-third
  International Conference on Machine Learning}}.
\newblock
\urldef\tempurl%
\url{https://openreview.net/forum?id=y1czGbUum5}
\showURL{%
\tempurl}


\end{thebibliography}

\appendix

\section{Worked example}
\label{sec:supp-example}

\begin{figure*}[t]
    \centering
    \begin{subfigure}[t]{\columnwidth}
        \centering
        \includegraphics{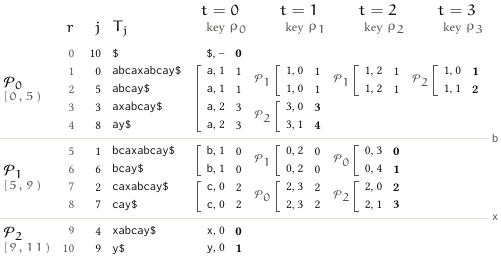}
        \caption{Ranks within partitions}
        \label{fig:example-ranks}
    \end{subfigure}\hfill
    \begin{subfigure}[t]{\columnwidth}
        \centering
        \includegraphics{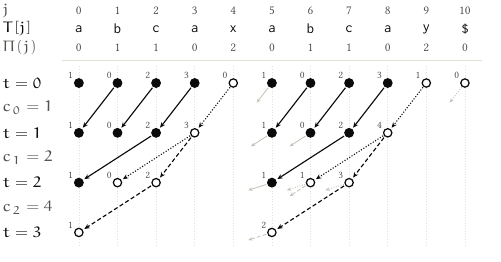}
        \caption{Messages between generations}
        \label{fig:example-messages}
    \end{subfigure}
    \caption{
        The worked example: $T = \texttt{abcaxabcay\$}$, with $\ell = K = 1$ and the splitters \texttt{b} and \texttt{x}.
        (a) The suffixes in \sa\ order, in partitions 0 to 2, labeled $\pset{0}$ to $\pset{2}$ at the left, above their ranges of \sa.
        Each generation lists the keys of its active suffixes, key~(\ref{eq:key0}) at \gen{0} and key~(\ref{eq:keyt}) after it, and the local ranks $\lrk{t}(j)$ that sorting them yields, which become the first components of the next generation's keys.
        A bracket covers the suffixes whose keys share their first component.
        After \gen{0}, the label to the left of a bracket names the partition that holds the bracket's extension suffixes.
        A rank is bold once it is final, and the suffix then leaves the columns; the dash is \sentinel's reserved second component~(\S\ref{sec:supp-terminator}).
        (b) The active suffixes of each generation in text order, with local ranks: filled if a suffix stays tied, hollow if it retires there.
        Arrows: every message that each generation sends, to the row below.
        Own messages of live ranks are solid, of pills dotted, and forwarded pills dashed.
        Black messages are read, one into each dot below the first row.
        Those that reach retired suffixes are faint stubs pointing toward them.
        No liveness filter is shown here for simplicity, and the filter would eliminate some of the dead messages.
    }
    \label{fig:example}
\end{figure*}

Let $T = \texttt{abcaxabcay\$}$, with $n = 11$, $\ell = K = 1$, and the splitters \texttt{b} and \texttt{x}.
Partition 0 holds the 5 suffixes that start with \sentinel\ or \texttt{a}, partition 1 the 4 that start with \texttt{b} or \texttt{c}, and partition 2 the 2 that start with \texttt{x} or \texttt{y}, so the three fill the ranges $[0, 5)$, $[5, 9)$, and $[9, 11)$ of \sa\ (\S\ref{sec:method-partition}).
$T$ holds two copies of \texttt{abca}, followed by \texttt{x} and \texttt{y}, so suffixes in the two copies stay tied until their contexts reach the \texttt{x} or the \texttt{y}.
Figure~\ref{fig:example} traces the run.

\para{Generation 0}
Each partition sorts its suffixes by key~(\ref{eq:key0}): here, by the first symbol, and then by the partition of the next suffix.
The four suffixes of partition 0 that start with \texttt{a} share that symbol, and the partition component splits them into two classes of tied suffixes, $\{0, 5\}$ and $\{3, 8\}$.
Suffixes 4, 9, and 10 have unique keys in their partitions and retire, while the other 8 stay tied.
Every suffix but 0 then sends its rank to the suffix before it: 10 messages, of which the three from 4, 9, and 10 are pills.

\para{Doubling generations}
At \gen{1}, each class compares the ranks of its members' extension suffixes, one position ahead.
Those of $\{3, 8\}$, 4 and 9, both lie in partition 2, as target coherence guarantees (Theorem~\ref{thm:tci}), and their ranks 0 and 1 split the class.
The members of each other class read the same rank, and stay tied.
At \gen{2}, whose extension suffixes lie $\cx{1} = 2$ positions ahead, $\{1, 6\}$ splits on the final ranks of 3 and 8, and $\{2, 7\}$ on those of 4 and 9; at \gen{3}, $\cx{2} = 4$ positions ahead, $\{0, 5\}$ splits on those of 4 and 9 as well.
Generations 1 to 3 thus rank 8, 6, and 2 suffixes.
Output assembly places each suffix at its rank past its partition's offset, which gives $\sa = [10, 0, 5, 3, 8, 1, 6, 2, 7, 4, 9]$.

\para{Messages}
Figure~\ref{fig:example-messages} draws all 25 messages that the run sends without the liveness filter.
The 16 that are read, in black, are one for each of the 8, 6, and 2 suffixes that generations 1 to 3 rank, as complete coverage states (Theorem~\ref{thm:coverage}): every dot below the first row has exactly one incoming black arrow.
The other 9, the faint stubs, reach suffixes that have retired, and are dead.
Forwarding passes a final rank on for as long as it is needed.
Suffix 4 retires at \gen{0}, but its rank is the extension rank of 3, 2, and 0 at generations 1, 2, and 3: 4 sends it to 3 as a pill, 3 forwards it to 2, and 2 forwards it to 0, and each of them retires on receiving it (\S\ref{sec:method-message}).
The rank of 9 serves 8, 7, and 5 in the same way.

\section{Proofs}
\label{sec:supp-proofs}

Notation and terminology are those of \S\ref{sec:method}.

\subsection{Target coherence}
\label{sec:supp-proof-tci}

\begin{theorem}[Theorem~\ref{thm:tci}, target coherence]
After every generation~$t$, any two suffixes tied within a partition share their first \cx{t} symbols, and their extension suffixes in \gen{t + 1} lie in one partition.
\end{theorem}
\begin{proof}
By induction on $t$.
After \gen{0}, the rank of a suffix is the number of suffixes in its partition with strictly smaller keys~(\ref{eq:key0}), so two suffixes tied after \gen{0} have equal keys.
They thus share their $\ell$-contexts, with $\cx{0} = \ell$, and the partitions $\pmap(j + \ell)$ of their extension suffixes in \gen{1}.

Let $t \ge 1$, and let $j \ne j'$ be tied after \gen{t} within a partition.
A shared rank is not final, so both were active in \gen{t}.
Equal ranks after \gen{t} mean equal keys~(\ref{eq:keyt}).
The $s$ suffixes tied at a rank $r$ after \gen{t - 1} receive new ranks in $[r, r + s)$, and no other suffix of the partition holds a rank in that range, so suffixes with different old ranks get different new ones.
Within one such group, a suffix's new rank exceeds $r$ by the number of suffixes in the group with strictly smaller extension ranks, so suffixes with different extension ranks get different new ones as well.
Hence $j$ and $j'$ were tied after \gen{t - 1}, and their extension suffixes $m = j + \cx{t - 1}$ and $m' = j' + \cx{t - 1}$ have equal ranks after \gen{t - 1}.

By the induction hypothesis for $j$ and $j'$, the suffixes $m$ and $m'$ lie in one partition, and $\ctx{j}{\cx{t - 1}} = \ctx{j'}{\cx{t - 1}}$.
As $j \ne j'$, also $m \ne m'$.
Two distinct suffixes of one partition share a rank only if it is not final, so $m$ and $m'$ were tied after \gen{t - 1}.
By the induction hypothesis for $m$ and $m'$, $\ctx{m}{\cx{t - 1}} = \ctx{m'}{\cx{t - 1}}$, and $m + \cx{t - 1}$ and $m' + \cx{t - 1}$ lie in one partition.
Concatenating the contexts of $j$ and $m$, and of $j'$ and $m'$, gives $\ctx{j}{\cx{t}} = \ctx{j'}{\cx{t}}$, as $\cx{t} = 2\cx{t - 1}$.
Finally, $m + \cx{t - 1} = j + \cx{t}$ and $m' + \cx{t - 1} = j' + \cx{t}$ are the extension suffixes of $j$ and $j'$ in \gen{t + 1}, and they lie in one partition.
\end{proof}

\subsection{Complete coverage}
\label{sec:supp-proof-coverage}

The proof uses pullback (\S\ref{sec:method-message}): if a suffix $j$ is active in \gen{t \ge 2}, then $j + \cx{t - 2}$ was active in \gen{t - 1}.

\begin{theorem}[Theorem~\ref{thm:coverage}, complete coverage]
In every \gen{t \ge 1}, each suffix receives at most one message, and each active suffix $j$ receives exactly one: it carries the rank of the extension suffix $m = j + \cx{t - 1}$, marked final if and only if $m$ retired before \gen{t}.
$m$ sends this message itself if it was active in \gen{t - 1}; otherwise, the midpoint $j + \cx{t - 2}$ forwards it.
\end{theorem}
\begin{proof}
By induction on $t$.
The messages that \gen{t} reads are those that \gen{t - 1} sends.
The liveness filter removes only messages to suffixes that have retired, so it removes none of the messages to an active suffix, and the claims hold with or without it.

For $t = 1$, every suffix $m \ge \ell$ was active in \gen{0}, and sent the message $\langle m - \ell, \lrk{0}(m), \mathit{final} \rangle$, with \textit{final} set if and only if $m$ retired in \gen{0}; \gen{0} forwards nothing.
So each suffix $j$ receives at most one message, from $m = j + \ell$, and an active suffix $j$ receives it.

Let $t \ge 2$, and let $z = j + \cx{t - 2}$, the midpoint of $j$ and $m$, as $\cx{t - 1} = 2\cx{t - 2}$.
In \gen{t - 1}, a suffix $x$ sends its own message to $x - \cx{t - 1}$, and forwards to $x - \cx{t - 2}$ the rank that it received for its extension suffix $x + \cx{t - 2}$.
So only $m$ can send its own message to $j$, and only $z$ can forward to $j$, a rank that $z$ received for $z + \cx{t - 2} = m$.

\emph{$m$ was active in \gen{t - 1}.}
Then $m$ sends $\lrk{t - 1}(m)$ to $j$, with \textit{final} set if and only if $m$ retired in \gen{t - 1}, \ie before \gen{t}.
A suffix that is not active in \gen{t - 1} forwards nothing in it.
If $z$ was active in \gen{t - 1}, it received there, by the induction hypothesis, the rank of $m$ marked final if and only if $m$ retired before \gen{t - 1}, which it did not; so $z$ forwards nothing.
Hence $j$ receives exactly one message.

\emph{$m$ was not active in \gen{t - 1}.}
Then $m$ retired before \gen{t - 1} and sends nothing, so $j$ receives at most the forward from $z$.
If $j$ is active in \gen{t}, then $z$ was active in \gen{t - 1} by pullback.
By the induction hypothesis, $z$ received there the rank of $m$ marked final, since $m$ retired before \gen{t - 1}.
So $z$ forwards it to $j$, marked final, and $j$ receives exactly one message.
The rank of $m$ is final, so it is the same in every later generation.
\end{proof}

\section{Design choices}
\label{sec:supp-design}

This section describes how \gpusaca\ handles the end of the text, and two alternatives that would free host memory of $T$ and \pmap: persisting the contexts, and relaying partition addresses.

\subsection{The end of the text}
\label{sec:supp-terminator}

Real corpora carry no sentinel, and the smallest value 0 is a valid token.
\gpusaca\ pads $T$ with $\max(\ell, K)$ instances of 0, so that every $\ell$-context and $K$-gram can be read whole, and orders every suffix before the longer suffixes that it prefixes, as the sentinel would.
The padding keeps \pmap\ monotone: the padded $K$-gram of a suffix shorter than $K$ is no larger than that of any longer suffix it prefixes.
The second component of key~(\ref{eq:key0}) is $\pmap(j + \ell) + \ell$ for a suffix $j$ with $j + \ell < n$.
A suffix whose $\ell$-context runs past the end, $j + \ell \ge n$, instead gets its length less one, $n - 1 - j < \ell$.
Two suffixes with equal padded $\ell$-contexts then sort correctly: if either runs past the end, the shorter is a prefix of the other, and has the smaller second component.
These reserved values are distinct and lie below every partition's, so each suffix of length at most $\ell$ has a unique key, and retires at \gen{0}.
By pullback (\S\ref{sec:method-message}), a suffix $j$ active after \gen{t} has $j + \cx{t - 1}$ active after \gen{t - 1}, so by induction every extension suffix of a doubling generation lies within $T$.

\subsection{Persisting the contexts}
\label{sec:supp-contexts}

After sampling the splitters, \gpusaca\ reads $T$ at random only in \gen{0}: to gather each suffix's $\ell$-context for key~(\ref{eq:key0}), and, without a resident \pmap, to find $\pmap(j \pm \ell)$ by binary-searching the $K$-grams at $j \pm \ell$ (\S\ref{sec:system-residency}).
Partitioning already reads $T$ and \pmap\ in index order, so its scatter can instead write $\ctx{j}{\ell}$ into the record of each suffix $j$, and $\pmap(j \pm \ell)$ too when \pmap\ is not resident, from a sliding window over the two; \gen{0} then reads them with the records.
$T$ is then read only sequentially, besides the $s$ positions that sampling reads, at most a few million on our corpora, which the drives serve directly.
So $T$ need never be resident; with the relay (\S\ref{sec:supp-relay}), \pmap\ need not be either.

\para{Cost}
The scatter's records, which cross \barr{0}, grow from a 5-byte suffix index to $5 + 2\ell = 37$ bytes per suffix at $\ell = 16$ for 2-byte tokens, and by $2 \lceil \log_2 p \rceil / 8$ more with the two addresses (\S\ref{sec:supp-relay}).
So the short-lived bytes at \barr{0} grow from $5n$ to about $37n$, twice the $18n$ at \barr{1}~(\S\ref{sec:method-residency}), and become the peak of the run.
Each of these bytes that spills is written and read once, to save the 2 bytes per suffix that $T$ takes in host memory.

\subsection{Relaying partition addresses}
\label{sec:supp-relay}

\gpusaca\ requires the partition map \pmap\ only to address messages (\S\ref{sec:method-message}).
The relay, sketched in \S\ref{sec:method-residency}, seeds the messages with the addresses instead and the mechanism sustains itself, so that \pmap\ need never be resident.

\begin{figure}[t]
    \centering
    \includegraphics{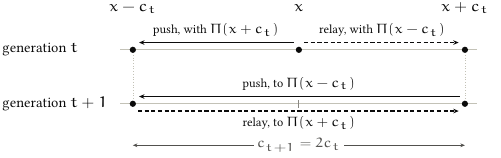}
    \caption{
        The relay mechanism.
        Pushes are solid, relays dashed.
        At \gen{t}, an active suffix $x$ knows the partitions of the suffixes $\cx{t}$ positions before and after it, and sends each of the two the other's partition.
        At \gen{t + 1}, the two are $\cx{t + 1}$ apart, and each sends a message to the other ($x - \cx{t}$ relays and $x + \cx{t}$ pushes), to the address that $x$ sent it.
    }
    \label{fig:relay}
\end{figure}

At \gen{t}, an active suffix $x$ sends its own message, its \emph{push}, to $\pmap(x - \cx{t})$; call it $x$'s \emph{push address}.
We call $\pmap(x + \cx{t})$ its \emph{relay address}.
The suffixes $x - \cx{t}$ and $x + \cx{t}$ are $2\cx{t} = \cx{t + 1}$ positions apart, so at \gen{t + 1}, the second pushes to the first, and the first \emph{relays} to the second: each needs the other's partition (Figure~\ref{fig:relay}).
With relaying, their midpoint $x$ provides both at \gen{t}: its push $\langle x - \cx{t}, \lrk{t}(x), \mathit{final}, \pmap(x + \cx{t}) \rangle$ carries its relay address~\footnote{
    The $\pmap(x + \cx{t})$ component is additional to the message structure (\S\ref{sec:method-message}) of the paper.
}, and its relay $\langle x + \cx{t}, \pmap(x - \cx{t}) \rangle$, sent to $\pmap(x + \cx{t})$, carries its push address.
In turn, $x$ received its push address from the relay of $x - \cx{t - 1}$, and its relay address from the push of its extension suffix $x + \cx{t - 1}$, both sent at \gen{t - 1}.
\Gen{0} seeds the relay mechanism with $\pmap(j \pm \ell)$ for each suffix $j$, computed by binary-searching $T[j \pm \ell, j \pm \ell + K)$ among the splitters (\S\ref{sec:method-partition}), or read from the records when $T$ is not resident (\S\ref{sec:supp-contexts}).
A pill forward of $x$ goes to its push address of \gen{t - 1}, which $x$ keeps with its record for one more generation.
Relays fill inboxes of their own, at most one per suffix, and join the active suffixes the same way as the pushes do, at the receive-side join (\S\ref{sec:method-message}).
The relayed addresses that the join attaches to the active suffixes replace the stagers' lookups into \pmap\ (\S\ref{sec:system-device}).

\para{No needed message is lost}
Only survivors relay; the retirees of a generation do not, as their relays would serve only messages that are never read (below).
A message is sent only if its address has arrived.
An address that arrives is exact, as the arithmetic above copies it from those of \gen{0}, by induction.
We argue that the sender of every message that is actually read has the reader's address.
Each case below rests on pullback (\S\ref{sec:method-message}): if a suffix is active after \gen{t}, its extension suffix, $\cx{t - 1}$ ahead, was active after \gen{t - 1}.
Let $x$ be active in \gen{t}.

\begin{enumerate}
\item \emph{Push to $y = x - \cx{t}$.}
If the recipient $y$ reads the push at \gen{t + 1}, it must have been active after \gen{t}.
So its extension suffix $y + \cx{t - 1} = x - \cx{t - 1}$, the midpoint of $y$ and $x$ (as $\cx{t} = 2\cx{t - 1}$), survived \gen{t - 1}, and relayed $\pmap(y)$ to $x$.

\item \emph{Relay to $x + \cx{t}$, if $x$ survives \gen{t}.}
Then its extension suffix $x + \cx{t - 1}$, the midpoint of $x$ and $x + \cx{t}$, survived \gen{t - 1}, and pushed to $x$ with $\pmap(x + \cx{t})$, rather than retiring and leaving $x$ a forward, which carries no address.
A retiree's relay is never needed: it would carry $\pmap(y)$ only for messages to $y$ that are read from \gen{t + 2} on, and $y$ is active after \gen{t + 1} only if its extension suffix of \gen{t + 1}, $y + \cx{t} = x$, survives \gen{t}.

\item \emph{Forward to $x - \cx{t - 1}$, if $x$ received a pill.}
It uses the push address of $x$ of the generation before.
If its recipient reads it, the recipient is active after \gen{t}, and thus after \gen{t - 1}, so the push case one generation earlier gave $x$ that address.
\end{enumerate}

By induction over the generations, every message that is read reaches its partition, and Theorem~\ref{thm:coverage} holds without \pmap; the pushes and forwards that are not sent would have been dead.

\para{Cost}
\gpusaca's implementation uses 5-byte suffix indices, and \pmap-entries take $\wpmap = (\lceil \log_2 p \rceil + 1) / 8$ bytes for $p$ partitions, with the liveness filter's flag (\S\ref{sec:method-message}), where the address is $\wpmap' = \lceil \log_2 p \rceil / 8$ bytes.
A run sends at most one relay per survivor and generation, consisting of a suffix index and an address.
Thus there are at most \cmass\ relays, where $\cmass = \bar{d}n$ is the doubling work (\S\ref{sec:evaluation-structure}): $0.15n$ to $4.5n$ on our corpora (Table~\ref{tab:datasets} of the paper), costing $(5 + \wpmap')\cmass$ bytes.
Each push and forward also carries an additional address, at most $n + \cmass + |\act{0}|$ of them over a run~(\S\ref{sec:method-message}), with \act{0} the suffixes still tied after \gen{0}, costing $(n + \cmass + |\act{0}|) \wpmap'$ bytes.
Each active record needs to carry an additional address, total \cmass\ over a run, costing $\cmass\wpmap'$ bytes.
The relay mechanism thus adds at most $5\cmass + (n + 3\cmass + |\act{0}|) \wpmap'$ bytes to a run's streams.
Each byte is written to SSD and read back once if it spills.
On \data{dolma-sample} ($p = 34$, $\wpmap' = 6/8$) and \data{github} ($p = 249$, $\wpmap' = 1$), that is $1.9n$ and $38n$ bytes, \ie \qty{15}{\giga\byte} and \qty{2.4}{\tera\byte}.
These streams can spill to SSD unlike \pmap.
But on our corpora, they outweigh \pmap's $\wpmap n$ bytes, about a byte per suffix, or \qty{7}{\giga\byte} and \qty{72}{\giga\byte} on the two above.
\gpusaca\ therefore keeps \pmap\ resident, skipping the additional SSD traffic (\S\ref{sec:method-residency}).
\section{Platform}
\label{sec:supp-platform}

Table~\ref{tab:hardware} lists the measured rates of the server that runs every experiment (\S\ref{sec:prelim-hw}).

\begin{table}[t]
\caption{
    The platform of all evaluations: an NVIDIA RTX PRO 6000 Blackwell GPU on PCIe~5.0~$\times$16, two 64-core Intel Xeon 6710E CPUs, and two Samsung PM9A3 NVMe SSDs.
    Per tier: capacity, sequential bandwidth, and random reads per second (8-byte loads by the GPU, by the GPU across the PCIe link; or by all 128 CPU cores; 4 KB pages per SSD).
}
\label{tab:hardware}
\small
\setlength{\tabcolsep}{3pt}
\begin{tabular}{@{}>{\raggedright\arraybackslash}p{0.22\columnwidth}>{\raggedright\arraybackslash}p{0.21\columnwidth}>{\raggedright\arraybackslash}p{0.33\columnwidth}>{\raggedright\arraybackslash}p{0.16\columnwidth}@{}}
    \toprule
    Tier & Capacity & Bandwidth & Random reads \\
    \midrule
    Device memory & \qty{96}{\giga\byte} & \qty{1.5}{\tera\byte\per\second} & 6.5--15\,G/s \\
    PCIe link & -- & \qty{56.5}{\giga\byte\per\second} to the device, \qty{57.3}{\giga\byte\per\second} back & 25--94\,M/s \\
    Host memory & \qty{1}{\tera\byte} & \qty{530}{\giga\byte\per\second} & 2.1\,G/s \\
    SSDs & $2 \times \qty{3.84}{\tera\byte}$ & \qty{6.7}{\giga\byte\per\second} read, \qty{4.1}{\giga\byte\per\second} write, each & 0.87\,M/s each \\
    \bottomrule
\end{tabular}
\end{table}

\section{Additional evaluation}
\label{sec:supp-experiments}

\subsection{Corpus structure}
\label{sec:supp-structure}

\para{Repeat depth and active sets}
Suffix $j$ stays tied after \gen{t} only if another suffix shares its $\cx{t}$-context, \ie if $\rdep(j) \ge \cx{t}$.
The partition component of key~(\ref{eq:key0}) breaks some of these ties up to $K$ symbols early (\S\ref{sec:method-tci}), so that the set \act{t} of suffixes still tied after \gen{t} satisfies
\begin{equation*}
    \{ j : \rdep(j) \ge \cx{t} + K \} \;\subseteq\; \act{t} \;\subseteq\; \{ j : \rdep(j) \ge \cx{t} \}.
\end{equation*}
Every active set in Figure~\ref{fig:structure-decay}a lies within this band, at most 8\% below its upper end: the component resolves 0.8--2.6 tokens of context beyond $\cx{t}$ for free, of the $K = 8$ it can.
Recomputing the tie-breaks from the suffix array and \pmap\ reproduces every active set that \gpusaca\ counts, exactly, on all corpora.

\para{Liveness filter}
Without the liveness filter (\S\ref{sec:method-message}), a run emits from 1.21 messages per token on \data{dolma-sample} to 6.18 on \data{github}, of which the dead ones, $n$ plus the forwards, total 1.06--1.68 per token on every corpus.
\Gen{0}'s flood holds 65--87\% of the dead messages, but only 14\% on \data{github}, where most suffixes survive it.
Figure~\ref{fig:structure-filter} ablates the filter: it skips 41--49\% of the dead messages of \gen{0}, whose recipients all retire in their senders' own generation, and 80--100\% of those of each later generation, most of whose recipients retired earlier.
In total, the filter skips \numrange{23}{48}\% of the messages and \numrange{46}{83}\% of the dead ones, without affecting correctness: in every generation of every corpus, the active sets equal those of the run without the filter, and the messages sent lie between the live ones and those emitted.
As it marks a retiree only once its partition has been ranked, what the filter skips depends on the order in which the partitions run, \ie the filtered message volume is not deterministic.

\begin{figure}[t]
    \centering
    \includegraphics{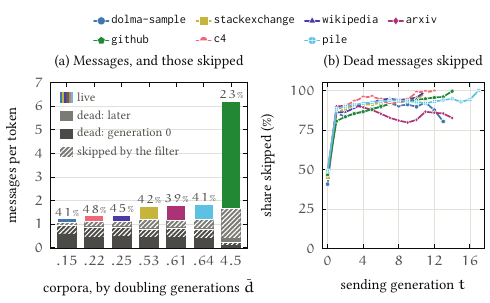}
    \caption{
        The liveness filter's ablation.
        (a) The messages per suffix that a run emits without the filter, with the corpora as in Figure~\ref{fig:structure-decay}b: the live ones (colored), and the dead ones of \gen{0} (dark) and of later generations, forwarded pills included; hatched, the share of each that the filter skips, and atop each bar, its share over the total message volume.
        (b) The share of each generation's dead messages that the filter skips, by the generation $t$ that sends them.
    }
    \label{fig:structure-filter}
\end{figure}

\para{Partition size skew}

The sampled splitters keep every partition close to its target size.
\gpusaca\ targets $p_{\min} = \lceil n / B \rceil$ partitions of $B$ suffixes, and fills them greedily from sampled $K$-grams so that each partition spans $h \approx 128 \log_2 n$ gaps between samples.
$B$ is set according to the device memory.
A partition's size then deviates from $B$ by about $1/\sqrt{h}$, and the largest of $p$ partitions by about $\sqrt{2 \ln p / h}$ (\S\ref{sec:method-partition}).
Sweeping $p$ from \num{34} to \num{24522} and $K$ from \num{4} to \num{32} tokens, the largest partition is found within \num{3}--\num{10}\% of $B$ and follows this envelope (Figure~\ref{fig:structure-skew}), fitting within the additional \num{20}\% space that \gpusaca\ reserves; the median partition is within \num{0.25}\% of $B$, and the greedy fill produces at most ten partitions beyond $p_{\min}$.
$K$ affects only feasibility: no splitter can divide a $K$-gram that outweighs $B$ (\S\ref{sec:method-partition}), which aborts \data{github} at $K = 4$, and both \data{github} and \data{pile} at $K = 8$ once $B$ falls to 32M suffixes ($p \ge 1984$ and $p \ge 7368$).

\begin{figure}[t]
    \centering
    \includegraphics{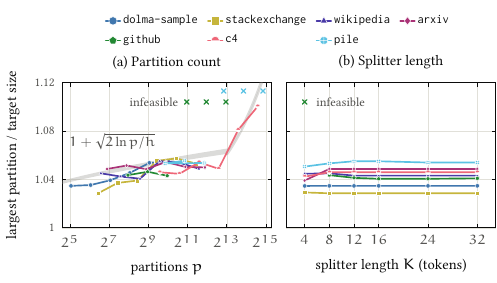}
    \caption{
        The largest partition-size as a multiple of its target size $B$, (a) against the partition count $p$, at $K = 8$ tokens, and (b) against the splitter length $K$, at each corpus' default $p$.
        The gray band is the sampling envelope $1 + \sqrt{2 \ln p / h}$ from the splitter design, with $h \approx 128 \log_2 n$ samples per partition; for \data{c4} and \data{pile}, $h$ falls as $1/p$ beyond $p \approx \num{8200}$, where the sample count reaches its cap.
    }
    \label{fig:structure-skew}
\end{figure}

\para{Heaviest $K$-grams}
Suffixes sharing their prefix $K$-gram share a partition, so the heaviest $K$-gram, of frequency $\fmax$, bounds the largest partition from below.
\gpusaca\ requires $\fmax \le B$, where the device holds partitions of up to $P_{\max} = 1.2B$ (\S\ref{sec:method-partition}).
Table~\ref{tab:supp-fmax} sets both against $B$ at each corpus' default partitioning.
At $K = 8$, the heaviest 8-gram occurs from \num{0.27} million times in \data{dolma-sample} to \num{54.6} million in \data{pile}: \num{0.1}--\num{21}\% of $B$, and at most \num{18}\% of $P_{\max}$.
Hence the sampling, not $\fmax$, sets the largest partition at \num{1.03}--\num{1.05}$B$.
Shorter splitters raise $\fmax$: at $K = 4$, the heaviest 4-gram of \data{github} holds \num{1.41}$B$ suffixes, and that of \data{c4} \num{0.74}$B$.
Of the \num{84} partitionings in the sweeps, the \num{77} whose heaviest $K$-gram holds at most \num{0.84}$B$ are feasible, and the \num{7} whose heaviest holds \num{1.41}--\num{6.75}$B$ abort.

Table~\ref{tab:supp-kgrams} decodes the heaviest 8-grams, none of which is prose.
In \data{stackexchange}, \data{github}, \data{c4}, and \data{pile}, it repeats a single token, a zero, a tab, a spaced period, or 16 spaces, and \numrange{66}{92}\% of its occurrences overlap the next one within longer runs.
Such runs give \data{github} and \data{pile} the heaviest 8-grams of all.
In the other three, it is boilerplate of some format: an empty field of YAML-encoded campaign-finance filings in \data{dolma-sample}, the headings that end a Wikipedia article, and the opening of a displayed equation in \data{arxiv}.

\begin{table}[!htb]
\caption{
    Each corpus' default partitioning: the partition count $p$, the target size $B$ in millions of suffixes, and the largest partition and the frequency $\fmax$ of the heaviest $K$-gram at $K = 8$ and $4$ tokens.
    The device holds partitions of up to $P_{\max} = 1.2B$.
}
\label{tab:supp-fmax}
\small
\setlength{\tabcolsep}{4pt}
\begin{tabular}{@{}l S[table-format=3.0] S[table-format=3.1] S[table-format=3.1] S[table-format=2.1] S[table-format=3.1]@{}}
    \toprule
    & & & {Largest} & \multicolumn{2}{c}{$\fmax$ (\% of $B$)} \\
    \cmidrule(l){5-6}
    Corpus & {$p$} & {$B$ (M)} & {(\% of $B$)} & {$K = 8$} & {$K = 4$} \\
    \midrule
    \data{dolma-sample}  & 34  & 243.6 & 103.5 & 0.1  & 0.4 \\
    \data{stackexchange} & 87  & 255.4 & 102.9 & 2.0  & 16.9 \\
    \data{wikipedia}     & 96  & 253.9 & 104.5 & 0.9  & 14.5 \\
    \data{arxiv}         & 119 & 256.1 & 104.9 & 2.3  & 30.7 \\
    \data{github}        & 249 & 257.8 & 104.4 & 17.8 & 140.9 \\
    \data{c4}            & 767 & 258.6 & 104.6 & 1.1  & 74.4 \\
    \data{pile}          & 922 & 258.9 & 105.3 & 21.1 & 44.7 \\
    \bottomrule
\end{tabular}

\bigskip
\caption{
    The heaviest 8-gram of each corpus, and its frequency $\fmax$ in millions; beneath it, where it occurs, as read off 40 of its occurrences drawn at random.
    \texttt{\mdseries\textvisiblespace} is a space, $\hookleftarrow$ a newline, and $\rightarrow$ a tab.
    The 8-grams are decoded with the corpora's tokenizers (\S\ref{sec:datasets}).
}
\label{tab:supp-kgrams}
\small
\setlength{\tabcolsep}{4pt}
\newcommand*{\ks}{\textvisiblespace}
\newcommand*{\kn}{\ensuremath{\hookleftarrow}}
\newcommand*{\kt}{\ensuremath{\rightarrow}}
\newcommand*{\kq}{\textquotesingle}
\begin{tabular}{@{}l S[table-format=2.2] l@{}}
    \toprule
    Corpus & {$\fmax$ (M)} & {Heaviest 8-gram; where it occurs} \\
    \midrule
    \data{dolma-sample}  & 0.27  & \texttt{\kn\ks\ks\ks\ks value:\ks\kq\kq\kn\ks\ks-} \\
                         &       & empty fields of YAML filing records \\[2pt]
    \data{stackexchange} & 5.19  & \texttt{00000000} \\
                         &       & zeros in hex dumps and numbers \\[2pt]
    \data{wikipedia}     & 2.18  & \texttt{\kn\kn References\kn\kn External\ks links} \\
                         &       & the headings that end an article \\[2pt]
    \data{arxiv}         & 5.89  & \texttt{\kn\textbackslash begin\{equation\}\kn\textbackslash} \\
                         &       & the start of displayed equations \\[2pt]
    \data{github}        & 45.83 & \texttt{\kt\kt\kt\kt\kt\kt\kt\kt} \\
                         &       & indentation of nested code \\[2pt]
    \data{c4}            & 2.95  & \texttt{\ks.\ks.\ks.\ks.\ks.\ks.\ks.\ks.} \\
                         &       & dot leaders of contents and listings \\[2pt]
    \data{pile}          & 54.65 & (\texttt{\ks}${} \times 16$)${} \times 8$ \\
                         &       & padding of plain-text tables \\
    \bottomrule
\end{tabular}
\end{table}

\subsection{Device memory}
\label{sec:supp-device}

In the sweep of Figure~\ref{fig:perf-device}, the device's peak memory is \qty{294}{\byte} per suffix of the largest partition: \qty{97}{\byte} for each of the three partitions in flight, one per CUDA stream, and \qty{2.7}{\byte} for the sort's workspace (\S\ref{sec:system-device}).
The largest partition stays 4--6\% above the mean throughout.
The heaviest 8-gram of \data{arxiv}, with 5.9M suffixes, sits well below even the 16M-suffix partitions at 16$\times$; that of \data{github}, with 45.8M, fits the 65M-suffix partitions at 4$\times$, but outweighs the 32M-suffix ones at 8$\times$, so the sweep splits \data{github} at $K = 20$ (\S\ref{sec:evaluation-reach}).
At 2$\times$, the builds are 9\% faster than at the device floor, partly as the largest partition sets the number of staging slabs (\S\ref{sec:system-pipeline}), and about half as many slabs pin in \qtyrange{10}{11}{\second} instead of \qtyrange{16}{18}{\second}.
Past 4$\times$, the builds slow by 12\% and 31\% at 8$\times$ and 16$\times$ on \data{arxiv}, and by 19\% and 46\% on \data{github}, as \gen{0} and doubling sort and rank up to 16 times as many partitions.
On \data{arxiv}, their costs rise from \qty{2.4}{} to \qty{4.0}{} and from \qty{1.2}{} to \qty{2.4}{\ns} per token; on \data{github}, where doubling dominates the build, from \qty{2.6}{} to \qty{3.7}{} and from \qty{8.0}{} to \qty{13.1}{\ns}.

\subsection{Bytes and time by phase}
\label{sec:supp-phases}

\para{Bytes}
Figure~\ref{fig:perf-bytes} splits the bytes that each build reads and writes across its phases (\S\ref{sec:evaluation-cost}).
We count them at the build's read and write calls: the text and the spill bypass the page cache, and the output either bypasses it or is synced as it is written (\S\ref{sec:system-layout}), so the calls carry all of the build's traffic to storage.

\para{Time}
This paragraph details the phase times of Figure~\ref{fig:perf-phases} beyond \S\ref{sec:evaluation-cost}.
Doubling costs \qtyrange{0.5}{1.7}{\ns} per token on every corpus but \data{github} (\qty{7.9}{\ns}), as their doubling work differs (\S\ref{sec:evaluation-structure}).
Per active suffix and generation, it costs more where generation~1 does: on \data{dolma-sample} (\qty{3.3}{\ns}), generation~1 joins 8.8 messages per active suffix, and on \data{c4} and \data{pile} (\qty{3.7}{} and \qty{2.6}{\ns}), it reads back \qty{1.0}{} and \qty{1.7}{\tera\byte} of spill, at \qty{8.0}{} and \qty{7.1}{\giga\byte\per\second}; on these three, the later generations cost \qtyrange{1.7}{2.5}{\ns}.
On \data{c4} and \data{pile}, partitioning costs \qty{0.7}{} and \qty{1.0}{\ns} per token, up from \qtyrange{0.5}{0.7}{\ns} in memory, as its scatter writes the spilled suffixes.
\data{github} spills less and later: its partitioning stays in host memory, and its generation~0, at \qty{2.6}{\ns} per token, spills \qty{0.38}{\tera\byte}, which generation~1 and output assembly read back.
Setup allocates the slabs that stage transfers between host and device, which take \qtyrange{14}{18}{\second} to pin in the background of partitioning's scatter: \data{dolma-sample} waits \qty{12}{\second} for the pinning, the next four corpora \qtyrange{2}{11}{\second}, and the scatters of \data{c4} and \data{pile} outlast it.

\begin{figure}[t]
    \centering
    \includegraphics{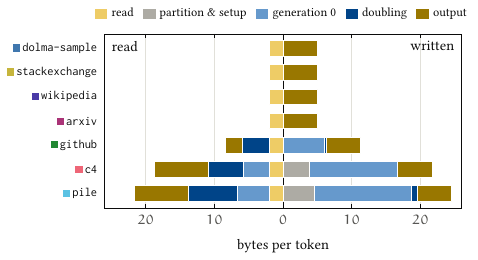}
    \caption{
        Bytes per token that each phase of \gpusaca's end-to-end build reads (left) and writes (right), counted at its read and write calls, in the phases of Figure~\ref{fig:perf-phases} of the paper.
        Every build reads its 2-byte tokens and writes the 5-byte suffix array; its other bytes are spill, which it reads back once.
    }
    \label{fig:perf-bytes}
\end{figure}

\section{Queries}
\label{sec:supp-query}

\subsection{Workloads}
\label{sec:supp-query-workloads}

Table~\ref{tab:supp-query-workloads} lists the 18 workloads of \S\ref{sec:evaluation-query}, each a batch that the query tool draws with its own seed, from a 64-bit Mersenne Twister.
A \emph{present} pattern starts at a uniformly random position of the corpus, so a batch draws frequent substrings more often, repeats some patterns, and may cross the separator between two documents.
A \emph{mutated} pattern is a present one with one position, chosen uniformly, replaced by a different token, drawn uniformly from all \num{65536} 16-bit values.
A \emph{random} pattern draws each of its tokens uniformly from the same \num{65536} values, which include the separator and the values outside the tokenizer's vocabulary.
The \emph{unigrams} are these \num{65536} values, once each.
A \emph{held-out} pattern is a present pattern of another corpus: of \data{wikipedia} for \data{stackexchange}, \data{arxiv}, \data{github}, \data{c4}, and \data{pile}; of \data{stackexchange} for \data{wikipedia}; and of arXiv, in the OLMo tokenizer of \data{dolma-sample}, for \data{dolma-sample}.
Under 1.6\% of the held-out patterns occur: except on \data{c4}, 4.9\% of the 10-token ones; and on \data{pile}, which contains Wikipedia: 19.5\% of the 10-token ones and 9.5\% of the 100-token ones.
Mutated patterns of 10 tokens or more, and random ones of 5 or more, almost never occur: these batches time the search for absent patterns.

\begin{table*}[t]
\caption{
    The workloads of \S\ref{sec:evaluation-query}: their pattern length $|P|$ in tokens, the pattern counts, the seed that draws them, and the share of their patterns that occur in each corpus.
}
\label{tab:supp-query-workloads}
\small
\begin{tabular}{@{}lrrrrrrrrrr@{}}
\toprule
& & & & \multicolumn{7}{c}{Share of patterns that occur} \\
\cmidrule(l){5-11}
Workload & $|P|$ & Patterns & Seed & \data{dolma-sample} & \data{stackexchange} & \data{wikipedia} & \data{arxiv} & \data{github} & \data{c4} & \data{pile} \\
\midrule
present-1 & 1 & \num{1000000} & 1 & 100\% & 100\% & 100\% & 100\% & 100\% & 100\% & 100\% \\
present-2 & 2 & \num{1000000} & 2 & 100\% & 100\% & 100\% & 100\% & 100\% & 100\% & 100\% \\
present-5 & 5 & \num{1000000} & 3 & 100\% & 100\% & 100\% & 100\% & 100\% & 100\% & 100\% \\
present-10 & 10 & \num{1000000} & 4 & 100\% & 100\% & 100\% & 100\% & 100\% & 100\% & 100\% \\
present-100 & 100 & \num{100000} & 5 & 100\% & 100\% & 100\% & 100\% & 100\% & 100\% & 100\% \\
present-1000 & 1000 & \num{100000} & 6 & 100\% & 100\% & 100\% & 100\% & 100\% & 100\% & 100\% \\
mutated-2 & 2 & \num{1000000} & 7 & 33.4\% & 22.6\% & 17.6\% & 17.9\% & 21.9\% & 27.3\% & 29.3\% \\
mutated-5 & 5 & \num{1000000} & 8 & 0.129\% & 0.41\% & 0.305\% & 0.355\% & 0.631\% & 0.556\% & 0.603\% \\
mutated-10 & 10 & \num{1000000} & 9 & $<$0.1\% & $<$0.1\% & $<$0.1\% & $<$0.1\% & $<$0.1\% & $<$0.1\% & $<$0.1\% \\
mutated-100 & 100 & \num{100000} & 10 & 0 & 0 & 0 & 0 & 0 & 0 & 0 \\
mutated-1000 & 1000 & \num{100000} & 11 & 0 & 0 & 0 & 0 & 0 & 0 & 0 \\
random-1 & 1 & \num{1000000} & 12 & 76.3\% & 48.5\% & 48.4\% & 47.8\% & 48.6\% & 48.4\% & 48.6\% \\
random-2 & 2 & \num{1000000} & 13 & 3.42\% & 1.6\% & 2.03\% & 0.982\% & 1.49\% & 4.07\% & 4.6\% \\
random-5 & 5 & \num{1000000} & 14 & 0 & 0 & 0 & 0 & 0 & 0 & 0 \\
random-10 & 10 & \num{1000000} & 15 & 0 & 0 & 0 & 0 & 0 & 0 & 0 \\
held-out-10 & 10 & \num{1000000} & 17 & 1.13\% & 0.39\% & 1.54\% & 0.346\% & 0.362\% & 4.87\% & 19.5\% \\
held-out-100 & 100 & \num{100000} & 18 & $<$0.1\% & $<$0.1\% & $<$0.1\% & $<$0.1\% & 0 & 0.268\% & 9.55\% \\
unigram & 1 & \num{65536} & -- & 76.3\% & 48.5\% & 48.4\% & 47.8\% & 48.6\% & 48.4\% & 48.7\% \\
\bottomrule
\end{tabular}
 \end{table*}

\subsection{Measurement}
\label{sec:supp-query-measurement}

\para{Indexes}
Each engine searches its own index of the same token array.
\ig\ orders the suffixes by their bytes, which hold each token little-endian; \gpusaca\ orders by the tokens.
So the two engines place an interval at different ranks: we compare their counts, per batch by a hash of all its counts.

\para{Timing}
Both engines time a batch from its patterns in host memory to their counts in host memory, after loading the index.
Each process times every batch six times in a row, three in placement D, and takes the median.
Cold repetitions in D first drop the engine's suffix array from the page cache; the hot ones follow the cold ones, and keep its pages.

\subsection{Results per workload}
\label{sec:supp-query-results}

Tables~\ref{tab:supp-query-memory} and~\ref{tab:supp-query-drive} give each engine's throughput per workload, and \gpusaca's speedup, with the index in memory and with \sa\ on SSD.
Relative to \ig's engine, \gpusaca's is slowest on the batches of 1000-token patterns, except cold on SSD: \numrange{2.0}{2.6}$\times$ in memory, and \numrange{3.4}{3.7}$\times$ hot.
Such a batch holds \qty{200}{\mega\byte} of patterns, which the engine copies to the device whole, in seven chunks that overlap the bracket search on the host: at the rate of Table~\ref{tab:hardware}, the copy alone takes at least \qty{3.5}{\ms} of the batch's \qtyrange{4.7}{7.3}{\ms} in placement C.
The bracket search also compares each pattern with the text over up to 1000 tokens, and the host searches the chunks for \numrange{52}{86}\% of the batch's time.

\begin{table*}[t]
\caption{
    Throughput (end to end) per workload, in millions of patterns per second, of \ig's engine and \gpusaca's, and \gpusaca's speedup ($\times$), with the index in memory.
    The last column gives each row's median over the workloads.
}
\label{tab:supp-query-memory}
\footnotesize
\setlength{\tabcolsep}{2.5pt}
\begin{tabular}{@{}lrrrrrrrrrrrrrrrrrrr@{}}
\toprule
& \multicolumn{6}{c}{present} & \multicolumn{5}{c}{mutated} & \multicolumn{4}{c}{random} & \multicolumn{2}{c}{held-out} & {unigram} & \\
\cmidrule(lr){2-7} \cmidrule(lr){8-12} \cmidrule(lr){13-16} \cmidrule(lr){17-18} \cmidrule(lr){19-19}
$|P|$ & 1 & 2 & 5 & 10 & 100 & 1000 & 2 & 5 & 10 & 100 & 1000 & 1 & 2 & 5 & 10 & 10 & 100 & 1 & {median} \\
\midrule
\multicolumn{20}{@{}l}{\data{dolma-sample} (A)} \\
\quad \ig & 29.2 & 12.6 & 14.5 & 15.4 & 9.92 & 9.13 & 17.2 & 17.1 & 16.5 & 10.0 & 9.46 & 23.2 & 23.1 & 23.2 & 22.9 & 23.8 & 12.4 & 8.75 & 16.0 \\
\quad \gpusaca & 653 & 457 & 473 & 474 & 138 & 20.1 & 698 & 631 & 540 & 151 & 21.6 & 781 & 950 & 922 & 778 & 658 & 179 & 533 & 537 \\
\quad $\times$ & 22.4 & 36.2 & 32.6 & 30.9 & 13.9 & 2.2 & 40.6 & 37.0 & 32.7 & 15.1 & 2.3 & 33.7 & 41.2 & 39.8 & 34.0 & 27.7 & 14.4 & 60.9 & 32.6 \\
\addlinespace
\multicolumn{20}{@{}l}{\data{stackexchange} (C)} \\
\quad \ig & 30.5 & 12.7 & 11.8 & 13.3 & 9.18 & 8.59 & 19.2 & 16.5 & 15.4 & 9.36 & 8.73 & 34.5 & 30.2 & 30.4 & 30.1 & 18.3 & 10.9 & 8.06 & 14.3 \\
\quad \gpusaca & 459 & 329 & 191 & 192 & 93.2 & 20.6 & 422 & 352 & 271 & 108 & 18.6 & 514 & 531 & 513 & 459 & 304 & 127 & 242 & 287 \\
\quad $\times$ & 15.0 & 25.9 & 16.1 & 14.4 & 10.2 & 2.4 & 22.0 & 21.4 & 17.6 & 11.5 & 2.1 & 14.9 & 17.6 & 16.9 & 15.2 & 16.6 & 11.6 & 30.0 & 15.7 \\
\addlinespace
\multicolumn{20}{@{}l}{\data{wikipedia} (C)} \\
\quad \ig & 25.3 & 10.8 & 12.2 & 13.5 & 9.22 & 8.58 & 18.3 & 16.2 & 15.2 & 9.43 & 8.76 & 30.5 & 27.4 & 27.5 & 27.3 & 19.4 & 11.3 & 9.33 & 14.3 \\
\quad \gpusaca & 423 & 287 & 222 & 229 & 99.5 & 18.8 & 425 & 367 & 296 & 106 & 22.4 & 505 & 515 & 505 & 447 & 340 & 127 & 237 & 291 \\
\quad $\times$ & 16.7 & 26.4 & 18.2 & 17.0 & 10.8 & 2.2 & 23.2 & 22.7 & 19.5 & 11.3 & 2.6 & 16.6 & 18.8 & 18.4 & 16.3 & 17.5 & 11.3 & 25.4 & 17.2 \\
\addlinespace
\multicolumn{20}{@{}l}{\data{arxiv} (C)} \\
\quad \ig & 31.5 & 13.6 & 11.4 & 13.2 & 9.14 & 8.46 & 20.0 & 16.2 & 15.0 & 9.26 & 8.50 & 36.1 & 32.2 & 32.3 & 32.0 & 19.2 & 11.5 & 8.17 & 14.3 \\
\quad \gpusaca & 465 & 358 & 218 & 232 & 101 & 20.8 & 442 & 387 & 308 & 112 & 21.6 & 522 & 549 & 523 & 468 & 343 & 128 & 243 & 325 \\
\quad $\times$ & 14.8 & 26.3 & 19.1 & 17.6 & 11.0 & 2.5 & 22.1 & 24.0 & 20.5 & 12.1 & 2.5 & 14.5 & 17.0 & 16.2 & 14.6 & 17.9 & 11.2 & 29.7 & 16.6 \\
\addlinespace
\multicolumn{20}{@{}l}{\data{github} (C)} \\
\quad \ig & 26.6 & 11.7 & 10.0 & 10.5 & 7.65 & 6.83 & 17.1 & 14.7 & 13.7 & 8.35 & 7.52 & 28.6 & 26.4 & 26.5 & 26.2 & 16.3 & 9.94 & 9.11 & 12.7 \\
\quad \gpusaca & 432 & 326 & 210 & 189 & 46.0 & 13.9 & 410 & 366 & 298 & 74.4 & 19.2 & 494 & 497 & 490 & 441 & 303 & 123 & 234 & 300 \\
\quad $\times$ & 16.2 & 27.8 & 21.0 & 18.1 & 6.0 & 2.0 & 24.0 & 24.9 & 21.8 & 8.9 & 2.6 & 17.2 & 18.8 & 18.5 & 16.9 & 18.6 & 12.4 & 25.7 & 18.3 \\
\bottomrule
\end{tabular}
 \end{table*}

\begin{table*}[t]
\caption{
    Throughput (end to end) per workload and \gpusaca's speedup, as in Table~\ref{tab:supp-query-memory}, with \sa\ on SSD (placement D), cold and hot, and the pages that each engine reads from the drive per pattern, cold.
}
\label{tab:supp-query-drive}
\footnotesize
\setlength{\tabcolsep}{2.5pt}
\begin{tabular}{@{}lrrrrrrrrrrrrrrrrrrr@{}}
\toprule
& \multicolumn{6}{c}{present} & \multicolumn{5}{c}{mutated} & \multicolumn{4}{c}{random} & \multicolumn{2}{c}{held-out} & {unigram} & \\
\cmidrule(lr){2-7} \cmidrule(lr){8-12} \cmidrule(lr){13-16} \cmidrule(lr){17-18} \cmidrule(lr){19-19}
$|P|$ & 1 & 2 & 5 & 10 & 100 & 1000 & 2 & 5 & 10 & 100 & 1000 & 1 & 2 & 5 & 10 & 10 & 100 & 1 & {median} \\
\midrule
\multicolumn{20}{@{}l}{\data{c4} (D)} \\
\quad \ig, cold & 1.56 & 0.087 & 0.054 & 0.057 & 0.040 & 0.040 & 0.11 & 0.069 & 0.062 & 0.040 & 0.040 & 1.40 & 0.21 & 0.22 & 0.22 & 0.10 & 0.055 & 0.092 & 0.078 \\
\quad \gpusaca, cold & 16.8 & 0.80 & 0.50 & 0.56 & 0.43 & 0.45 & 1.16 & 0.76 & 0.64 & 0.44 & 0.45 & 14.4 & 2.69 & 2.71 & 2.72 & 0.81 & 0.52 & 1.00 & 0.78 \\
\quad $\times$, cold & 10.8 & 9.1 & 9.2 & 9.8 & 10.9 & 11.3 & 10.6 & 11.0 & 10.3 & 11.0 & 11.2 & 10.3 & 12.5 & 12.5 & 12.5 & 7.9 & 9.5 & 10.9 & 10.8 \\
\addlinespace
\quad \ig, hot & 12.0 & 4.80 & 5.45 & 6.16 & 4.72 & 4.64 & 8.03 & 7.20 & 6.78 & 4.72 & 4.63 & 12.5 & 13.2 & 13.4 & 13.3 & 7.84 & 5.50 & 6.50 & 6.64 \\
\quad \gpusaca, hot & 180 & 86.2 & 81.9 & 88.3 & 48.0 & 15.9 & 142 & 124 & 103 & 51.1 & 16.6 & 204 & 203 & 206 & 205 & 119 & 64.9 & 129 & 111 \\
\quad $\times$, hot & 15.0 & 18.0 & 15.0 & 14.3 & 10.2 & 3.4 & 17.7 & 17.2 & 15.2 & 10.8 & 3.6 & 16.3 & 15.3 & 15.3 & 15.4 & 15.2 & 11.8 & 19.8 & 15.2 \\
\addlinespace
\quad pages, \ig & 0.28 & 5.27 & 8.54 & 8.05 & 11.3 & 11.3 & 4.17 & 6.66 & 7.44 & 11.3 & 11.3 & 0.30 & 2.13 & 2.12 & 2.12 & 4.43 & 8.20 & 4.52 & 5.96 \\
\quad pages, \gpusaca & 0.03 & 0.84 & 1.40 & 1.18 & 1.19 & 1.18 & 0.58 & 0.84 & 1.01 & 1.18 & 1.18 & 0.04 & 0.22 & 0.22 & 0.22 & 0.84 & 1.05 & 0.55 & 0.84 \\
\addlinespace
\multicolumn{20}{@{}l}{\data{pile} (D)} \\
\quad \ig, cold & 1.33 & 0.081 & 0.051 & 0.055 & 0.039 & 0.039 & 0.10 & 0.067 & 0.060 & 0.039 & 0.039 & 1.23 & 0.18 & 0.18 & 0.18 & 0.084 & 0.048 & 0.081 & 0.074 \\
\quad \gpusaca, cold & 15.1 & 0.72 & 0.46 & 0.51 & 0.43 & 0.44 & 1.06 & 0.71 & 0.60 & 0.43 & 0.44 & 14.1 & 2.34 & 2.37 & 2.36 & 0.65 & 0.48 & 0.96 & 0.68 \\
\quad $\times$, cold & 11.3 & 9.0 & 8.9 & 9.4 & 11.1 & 11.3 & 10.5 & 10.7 & 10.1 & 11.1 & 11.4 & 11.5 & 13.2 & 13.2 & 13.2 & 7.8 & 9.9 & 11.9 & 11.1 \\
\addlinespace
\quad \ig, hot & 11.6 & 4.79 & 5.20 & 5.95 & 4.58 & 4.45 & 7.77 & 7.05 & 6.69 & 4.75 & 4.58 & 11.5 & 12.5 & 12.7 & 12.6 & 7.26 & 5.02 & 6.19 & 6.44 \\
\quad \gpusaca, hot & 165 & 78.2 & 69.3 & 83.1 & 46.3 & 15.5 & 131 & 119 & 101 & 42.8 & 16.8 & 186 & 186 & 202 & 184 & 110 & 56.7 & 132 & 106 \\
\quad $\times$, hot & 14.3 & 16.3 & 13.3 & 14.0 & 10.1 & 3.5 & 16.8 & 16.8 & 15.2 & 9.0 & 3.7 & 16.2 & 14.9 & 15.9 & 14.7 & 15.2 & 11.3 & 21.4 & 14.8 \\
\addlinespace
\quad pages, \ig & 0.33 & 5.68 & 8.95 & 8.38 & 11.6 & 11.6 & 4.58 & 6.90 & 7.70 & 11.6 & 11.6 & 0.34 & 2.59 & 2.57 & 2.57 & 5.48 & 9.43 & 5.21 & 6.29 \\
\quad pages, \gpusaca & 0.03 & 0.88 & 1.43 & 1.21 & 1.22 & 1.21 & 0.63 & 0.85 & 1.01 & 1.19 & 1.21 & 0.04 & 0.26 & 0.25 & 0.25 & 1.02 & 1.13 & 0.56 & 0.95 \\
\bottomrule
\end{tabular}
 \end{table*}

\subsection{Placements compared}
\label{sec:supp-query-placements}

Table~\ref{tab:supp-query-plans} also runs three corpora in their other placements that fit: \data{dolma-sample} in C, B, and D; \data{arxiv} in B and D; and \data{github} in D, as its text does not fit the device.
The placement that our engine picks is the fastest on each corpus.

\begin{table*}[t]
\centering
\caption{
    \gpusaca's median throughput over the workloads, in millions of patterns per second, and its median speedup over \ig's engine ($\times$), with each corpus in each placement that fits it; in bold, the placement that \gpusaca's engine picks.
}
\label{tab:supp-query-plans}
\small
\setlength{\tabcolsep}{4pt}
\begin{tabular}{@{}lrrrrrr@{}}
\toprule
& \multicolumn{2}{c}{\data{dolma-sample}} & \multicolumn{2}{c}{\data{arxiv}} & \multicolumn{2}{c}{\data{github}} \\
\cmidrule(l){2-3} \cmidrule(l){4-5} \cmidrule(l){6-7}
Placement & {M/s} & {$\times$} & {M/s} & {$\times$} & {M/s} & {$\times$} \\
\midrule
A & \textbf{537} & \textbf{32.6} & {--} & {--} & {--} & {--} \\
C & 296 & 17.1 & \textbf{325} & \textbf{16.6} & \textbf{300} & \textbf{18.3} \\
B & 70.9 & 4.5 & 64.4 & 4.3 & {--} & {--} \\
D, cold & 0.79 & 5.8 & 1.16 & 8.3 & 1.01 & 8.2 \\
D, hot & 190 & 15.3 & 190 & 15.2 & 176 & 17.3 \\
\bottomrule
\end{tabular}
 \end{table*}

\para{B against C}
B is the slowest placement in memory: C's throughput is 4.2 and 5.1 times as high on \data{dolma-sample} and \data{arxiv}.
In B, the device's bracket search, with its copies of brackets across PCIe, takes 89\% and 96\% of a batch of 10-token present patterns on these corpora: \qty{14.2}{} and \qty{43.0}{\ms} per million patterns, with \qty{40}{} and \qty{151}{\giga\byte} of suffix array.
The host's search in C takes \qtyrange{3.4}{4.6}{\ms} per million on all five corpora that C holds, of 8 to 64 billion tokens.
The device reads pinned host memory at random at 51M reads per second at \qty{16}{\gibi\byte}, and at fewer as the memory grows, 32M at \qty{140}{\gibi\byte} and 25M at \qty{300}{\gibi\byte}, while it copies \qty{56.5}{\giga\byte\per\second}: each read waits on the translation of a host address, not on bandwidth.
The host's 128 cores, with 16 loads in flight each, reach \num{2062}M random loads per second, 41 to 64 times as many at the sizes of these two suffix arrays.

\para{D against memory}
Hot, placement D reaches 0.58 to 0.64 times C's throughput: the page cache maps the suffix array in pages of \qty{4}{\kibi\byte}, while the heap's pages are huge, and with \qty{4}{\kibi\byte} pages, the host's random loads fall to \num{1114}M per second, 0.54 times as many.
Cold, D's lead over \ig's engine grows with the index: \num{5.8}$\times$ on \data{dolma-sample}, \numrange{8.2}{8.3}$\times$ on \data{arxiv} and \data{github}, and \numrange{10.8}{11.1}$\times$ on \data{c4} and \data{pile}, as \ig's reads per present 10-token pattern grow from 3.5 pages to 8.4, while ours stay at 1.0 to 1.2.

\section{Validation}
\label{sec:supp-validation}

We validate \gpusaca's output, and test its implementation, in the following ways.
This list is not exhaustive.
\begin{itemize}
    \item \textbf{Suffix arrays.}
    An independent validator checks the suffix array that \gpusaca\ builds for every corpus: that it is a permutation of the suffixes, and that every pair of adjacent suffixes is in lexicographic order, compared in full.

    \item \textbf{Structural cross-check.}
    From the suffix array and the run's partitioning alone, the same tool recomputes which suffixes stay tied after each generation, and reproduces every active set that the run counts, exactly, on all corpora (\S\ref{sec:supp-experiments}).

    \item \textbf{Invariant checks.}
    Test builds check the method's invariants inline, as the construction runs: \pmap\ against a reference partition map computed on the host (\S\ref{sec:method-partition}); every suffix in its partition, exactly once; target coherence among tied suffixes (Theorem~\ref{thm:tci}); complete coverage of messages (Theorem~\ref{thm:coverage}); the routing of every message, and each inbox's size and checksum against its senders'; and the retirement flags against the active sets.
    They also validate their own output suffix array at the end.

    \item \textbf{Placement stress.}
    Test builds also move chunks between host memory and SSD at every barrier, refuse charges to the host budget at random, force every stream to SSD, and shrink the chunks so that small inputs span many of them, to exercise every placement path.

    \item \textbf{Queries.}
    The query engine checks its index at load, and can check every answer (\S\ref{sec:system-query}).
\end{itemize}
 
\end{document}